\documentclass[10pt]{article}
\usepackage{amsfonts,amsmath,latexsym}
\usepackage[T1]{fontenc}
\usepackage[dvips]{graphicx}
\usepackage{epsfig}
\usepackage{enumerate}
\usepackage{epsf}
\usepackage{amsthm}
\usepackage{color}
\usepackage{amssymb}
\usepackage{fancyhdr} 
\newtheorem{thm}{Theorem}[section]
\newtheorem{propo}[thm]{Proposition}
\newtheorem{lemme}[thm]{Lemma}
\newtheorem{corro}[thm]{Corollary}
\newtheorem{defi}[thm]{Definition}
\newtheorem{remarque}[thm]{Remark}
\newtheorem{example}[thm]{Example}
\newtheorem{Hyp}{Assumption}

\newenvironment{prooff}[1]{\begin{trivlist}
\item {\it  Proof}\quad} {\qed\end{trivlist}}

\def\R{\mathbb R}

\def\C{\mathbb P}
\def\C{\mathbb C}
\def\E{\mathbb E}
\def\P{\mathbb P}

\def\shd{{\cal D}}

\def\shf{{\cal F}}

\def\shl{{\cal L}}

\def\halb{{\frac{1}{2}}}

\author{{\sc St\'ephane GOUTTE $^*$}\footnote{Universit\'e Paris 13,
    Math\'ematiques LAGA, Institut Galil\'ee, 99 Av. J.B. Cl\'ement 93430
    Villetaneuse. E-mail:{\tt
      goutte@math.univ-paris-diderot.fr}} \thanks{Luiss Guido Carli - Libera
    Universit\`a Internazionale degli Studi Sociali Guido Carli di Roma}
  {\sc,}\ {\sc Nadia OUDJANE $^*$}\thanks{EDF R\&D, 
Universit\'e Paris 13, FiME (Laboratoire de Finance des March\'es de l'Energie
(Dauphine, CREST,  EDF R\&D) www.fime-lab.org).
E-mail:{\tt 
nadia.oudjane@edf.fr}}
\ {\sc and}\ {\sc Francesco RUSSO }  
\footnote{ENSTA ParisTech, UMA, 
 Unit\'e de Math\'ematiques appliqu\'ees,
828, boulevard des Mar\'echaux,
F-91120 Palaiseau (France)
}
\thanks{INRIA Rocquencourt 
and Cermics Ecole des Ponts, Projet MATHFI.
  E-mail:{\tt  francesco.russo@ensta-paristech.fr}}
}

\date{December 3rd 2012}
\title{\bf Variance optimal hedging for continuous time additive processes 
         and applications}

\newcommand{\MBFigure}[6]{
$\left. \right.$ \\
\refstepcounter{figure}
\addcontentsline{lof}{figure}{\numberline{\thefigure}{\ignorespaces #5}}
\begin{center}
\begin{minipage}{#1cm}
\centerline{\includegraphics[width=#2cm,angle=#3]{#4}}
\begin{center}
\upshape{F\textsc{ig} \normal
\end{center}
size{\thefigure}. $-$} #5
\end{center}
\label{#6}
\end{minipage}
\end{center}
$\left. \right.$ \\}

\begin{document}
\maketitle

\begin{abstract} For a large class of vanilla contingent claims, 
we establish an explicit F\"ollmer-Schweizer 
decomposition when the underlying is 
an exponential of an additive process.
This allows to provide an efficient algorithm for solving the
mean variance hedging problem.
Applications to models derived from the electricity market are performed.
\end{abstract}

\bigskip\noindent {\bf Key words and phrases:}  Variance-optimal
hedging, F\"ollmer-Schweizer decomposition, L\'evy's processes, 
Electricity markets, Processes
 with independent increments, Additive processes.

\bigskip\noindent  {\bf 2000  AMS-classification}: 60G51, 60H05, 60J25, 60J75


\bigskip




\section{Introduction}

\setcounter{equation}{0}

There are basically two main approaches to define the
 \textit{mark to market} of a contingent claim: one relying on the
 \textit{no-arbitrage assumption} and the other related to a
 \textit{hedging portfolio}, those two approaches converging in the
 specific case of complete markets.
In this paper we focus on the hedging approach.
 A simple introduction to the different hedging and pricing models in incomplete markets can be found in chapter~10 of~\cite{livreTankovCont}.\\
When the market is not complete,  it is not possible, in general, to
hedge perfectly an option.
 One has to specify  risk criteria,  and
 consider the hedging strategy that minimizes the distance (in terms of
 the given criteria) between the payoff of the option and the terminal
 value of the hedging portfolio.
In practice  the price of the option is
related to two components: first, the initial-capital value
and  second the quantitative evaluation of
 the residual risk induced by this imperfect hedging strategy
 (due to incompleteness). \\
%
Several criteria can be adopted. 
The aim of super-hedging is to hedge all cases. This approach yields in general prices that are too  expensive to be realistic~\cite{elkaroui-quenez}. 
Quantile hedging modifies this approach allowing for a limited probability of loss~\cite{fl99}. 
Indifference utility pricing introduced in~\cite{hodges-neuberger} defines the price of an option to sell (resp. to buy) as the minimum initial value
 s.t. the hedging portfolio with the option sold (resp. bought) is equivalent
 (in term of utility) to the initial portfolio. 
Global quadratic hedging approach was developed by M. Schweizer
(\cite{S94},~\cite{S95bis}):
  the  distance defined by the expectation of the square 
of the difference
 between the hedging portfolio and the payoff is 
minimized. Then, contrarily to the case of utility maximization,
in general that approach  provides
linear prices  and hedge ratios with respect to the payoff. 


%
In this paper, we follow this last approach
either to derive the hedging strategy minimizing the
 \textit{global quadratic hedging error} 
 for a given 
 initial capital, 
or to derive both the initial capital and the hedging strategy
 minimizing the same error.
Both actions are referred to the objective measure.
Moreover we also derive explicit formulae for
the global quadratic hedging error which together with the initial
capital allows the practitioner to define his option price.

We spend now some words related to the global quadratic hedging approach
 which is also 
called {\it mean-variance hedging} or {\it global risk minimization}.
Given a square integrable r.v. $H$, we say that the pair $(V_0, \varphi)$
is optimal if $(c, v) = (V_0, \varphi)$ minimizes the functional
$\E \left (H-c- \int_0^T v dS\right)^2$. The quantity
 $V_0$  and process  $\varphi$ represent 
 the initial capital and the optimal hedging strategy
of the contingent claim $H$.

Technically speaking, the global risk minimization problem  
 is based on the local risk minimization one which is
 strictly related to  the so-called {\it F\"ollmer-Schweizer} decomposition 
(or FS decomposition) of a square integrable random variable
 (representing the contingent claim)
with respect to an $(\shf_t)$-semimartingale $S = M + A$  
modeling the asset price: $M$ is an $(\shf_t)$-local martingale
and $A$ is a bounded variation process with $A_0 = 0$.
 Mathematically, the FS decomposition, constitutes the generalization of the
 martingale representation theorem (Kunita-Watanabe representation), which is valid 
when $S$ is a Brownian motion or a martingale. Given a square integrable 
random variable $H$, the problem consists in expressing 
$H$ as $H_0 + \int_0^T \xi dS + L_T$ where $\xi$ is predictable and
$L_T$ is the terminal value of an orthogonal martingale $L$ to $M$,
i.e. the martingale part of $S$.
In the seminal paper  \cite{FS91},  the problem is treated 
for an  underlying process  $S$  with continuous paths.
In the general case, $S$ is said to satisfy the {\bf structure condition}
(SC)  if there is a predictable process $\alpha$ such that 
$A_t = \int_0^t \alpha_s d\langle M\rangle_s$ and 
$\int_0^T \alpha_s^2 d\langle M\rangle_s < \infty $ a.s.
In the sequel, most of the contributions were produced
in the multidimensional case. Here, for simplicity, we will
formulate all the results in the one-dimensional case. 

 $H_0$ constitutes in fact the initial capital and it is given by
 the expectation of $H$ under
the so called {\it variance optimal signed measure} (VOM). 
 Hence, in full generality, the initial capital $V_0$ is not guaranteed to be
 an arbitrage-free price. For continuous processes, the variance optimal
 measure 
is proved to be non-negative under  a  mild no-arbitrage
condition~\cite{S96}.
 Arai (\cite{Arai052} and~\cite{Arai05}) provides sufficient conditions
 for the variance-optimal martingale measure to be a probability
 measure, even 
for discontinuous semimartingales.\\
In the framework of FS decomposition, a
 process which plays a significant role is the so-called
{\it mean variance trade-off} (MVT) process $K$. This notion is 
inspired by the theory in
discrete time started by \cite{Schal94}; under condition (SC), in the continuous time case
$K$ is defined as $K_t = \int_0^t \alpha^2_s d\langle M\rangle_s, \ t \in [0,T]$.
 In fact, in \cite{S94} also appear a slight more general  condition, called  (ESC), together with a corresponding  EMVT process;   we will
nevertheless not discuss here further  details.
If the MVT process is deterministic, \cite{S94} solves the mean-variance hedging problem 
and also provides  an efficient 
relation between  the solution of the global risk minimization problem and 
 the FS decomposition, see Theorem~\ref{mainthm}.
We remark that, in the  continuous case, treated by \cite{FS91},
 no need of any condition on $K$ is required.  It also shows that, for obtaining
 the mentioned relation, previous 
condition is not far from being optimal.   
The next important step was done in 
\cite{MS95} where, under the only condition that $K$ is uniformly bounded,
the FS decomposition of any square integrable random variable 
exists, it is unique and the global minimization
problem admits a solution. \\
More recently has appeared an incredible amount of papers in
the framework of global (resp. local) risk minimization, so that 
it is impossible 
to list all of them and it is beyond our scope. 
Four significant papers containing a good list of references
are \cite{S01}, \cite{nunno}, \cite{CK08} and \cite{S10}.

In this paper, 
we are not interested in generalizing the conditions
under which the FS decomposition exists.
 The present article aims, in the spirit of a simplified 
 Clark-Ocone formula, at providing  an explicit form for the FS
 decomposition for a large class of European payoffs $H$, when the process $S$ 
is  an exponential of
 additive process which is not necessarily a martingale. 
%
From a practical point of view, this serves to compute efficiently 
the variance optimal hedging strategy which is directly related to 
the FS decomposition, since the mean-variance trade-off
 is for that type of processes deterministic. 
%
One major idea proposed by Hubalek, Kallsen and Krawczyk in
 \cite{Ka06}, in the case where the log price is a
L\'evy process, consists in determining an explicit 
expression for the  variance optimal hedging strategy
for exponential payoffs and then deriving, by linear combination
the corresponding optimal strategy for a large class of payoff functions
(through Laplace type transform).
Using the same idea, this paper  extends results of~\cite{Ka06} 
considering prices that are exponential of additive processes  and  
contingent claims that are Laplace-Fourier transform of a
finite measure.   In this generalized framework, 
we could  formulate assumptions as general as possible.  
In particular, 
our results do not require any assumption 
on the
 absolute continuity of the cumulant generating function of
 $\log(S_t)$, thanks to the use of a natural {\it reference variance
   measure}   instead of the usual Lebesgue measure, see Section
\ref{SSRVM}. 
In the context of non stationary processes, the idea to represent payoffs
 functions as Laplace transforms  was applied by  \cite{kluge} 
(that we discovered after
finishing our paper)   to derive explicit pricing formulae  and by 
\cite{filipovic}  to investigate time inhomogeneous affine processes.
However, the \cite{kluge} generalization
was limited to additive processes with absolutely continuous characteristics
and to the pricing application: hedging strategies were not addressed.

One practical motivation for considering processes with independent and 
possibly non stationary increments came from 
hedging problems in the electricity market. Because of non-storability of electricity, the hedging instrument 
is in that case, 
a forward contract with value $S^{0}_t = e^{-r(T_d-t)}(F_t^{T_d} - F_0^{T_d})$ where $F_t^{T_d}$ 
is the forward price given at time $t\leq T_d$ for delivery of 1MWh at time $T_d$. Hence, the dynamics of the underlying $S^0$ is directly related to the dynamics of forward prices.  Now, forward prices are known to exhibit both heavy tails (especially on the short term) and a volatility term structure according to the \textit{Samuelson hypothesis}~\cite{Sam65}. More precisely, 
as the delivery date $T_d$ approaches, the forward price is more sensitive to the information arrival concerning the electricity supply-demand balance for the given delivery date. This phenomenon causes great variations in the forward prices close to delivery and then increases the volatility. Hence, those features require the use of forward prices models with both non Gaussian and non stationary increments in the stream of the model proposed 
by Benth and Saltyte-Benth, see
\cite{BS03} and also \cite{Benth-book}.

The paper is organized as follows. After this introduction 
we introduce the notion of FS decomposition and 
describe  global risk minimization. 
 Then, we examine at Section~3  the explicit FS decomposition 
 for exponential of additive processes. Section~4 is devoted to the
solution to the global minimization problem, Section~5 to 
theoretical examples and Section 6 to
the case of a model intervening in the electricity market. Section 7 
is devoted to simulations. 




\section{Preliminaries on additive processes and
F\"ollmer-Schweizer decomposition}
\label{Generalities}

\setcounter{equation}{0}

In the whole paper, $T>0$, will be a fixed terminal time and we will 
denote by $(\Omega,\mathcal{F},(\mathcal{F}_t)_{t \in [0,T]},\P)$ a
 filtered probability space,
fulfilling the usual conditions.
 In the whole paper, without restriction of generality 
$\shf$ will stand for the $\sigma$-field $\shf_T$.

\subsection{Generating functions}
Let $X=(X_t)_{t \in [0,T]}$ be a real valued stochastic process.
%

\begin{defi}
\label{defKappa}
The {\bf cumulant generating function} of (the law of) $X_t$ is the
mapping $z \mapsto {\rm Log} (\mathbb{E}[e^{zX_t}])$ 
where ${\rm Log}(w) = \log(\vert w \vert ) + i {\rm Arg (w)}$
where ${\rm Arg (w)}$ is the Argument of $w$, chosen in $]-\pi,
\pi]$; ${\rm Log}$ is the principal value logarithm.
In particular we have
$$
\kappa_{X_t}:D \rightarrow \mathbb{C} \quad \textrm{with} \quad e^{\kappa_{X_t}(z)}=\mathbb{E}[e^{zX_t}]\ ,
$$
where $D:=\{z \in \mathbb{C}\ \vert\ \mathbb{E}[e^{Re(z)X_t}]<\infty, \ 
 \forall t \in [0,T]\}$.
 In the sequel, when there will be no ambiguity on the underlying
 process
 $X$, we will use the shortened notations $\kappa_t$ for
 $\kappa_{X_t}$.
We observe that $D$ includes the imaginary axis.
\end{defi}
\begin{remarque}\label{remarkR2}
\begin{enumerate}
\item For all $z \in D$, $\kappa_t(\bar{z})=\overline{\kappa_t(z)}\ ,$ where   $\bar z$ denotes the conjugate complex of $z\in \mathbb{C}$.
\item For all $z \in D\cap\mathbb{R}\ ,\ \kappa_t(z) \in \mathbb{R}\ .$
\end{enumerate}
\end{remarque}
In the whole paper $\R^\star$ will stand for $\R - \{0\}.$

\subsection{Semimartingales}
An $(\mathcal{F}_t)$-semimartingale $X=(X_t)_{t \in [0,T]}$ is a process
of the form $X=M+A$, where $M$ is an $(\mathcal{F}_t)$-local martingale
and $A$ is a bounded variation adapted process vanishing at
zero. $||A||_T$ will denote the total variation of $A$ on $[0,T]$.
If $A$ is $(\mathcal{F}_t)$-predictable then $X$ is called an 
 $(\mathcal{F}_t)$-special semimartingale.
The decomposition of an  $(\mathcal{F}_t)$-special semimartingale
is unique, see Definition 4.22 of \cite{JS03}.
 Given two  $(\mathcal{F}_t)$-
 locally square integrable martingales $M$ and $N$,
 $\left\langle M,N\right\rangle$
 will denote the angle bracket of $M$ and $N$, i.e. the 
unique bounded variation predictable process vanishing at zero such that $MN-\left\langle M,N\right\rangle$ is an $(\mathcal{F}_t)$-local martingale. If $X$ and $Y$ are $(\mathcal{F}_t)$-semimartingales, $\left[X,Y\right]$ denotes the square bracket of $X$ and $Y$, i.e. the quadratic covariation of $X$ and $Y$.  
In the sequel, if there is no confusion about the underlying
 filtration $(\shf_t)$, we will simply speak about
semimartingales, special semimartingales, local martingales, martingales.

 All along this paper we will consider $\mathbb{C}$-valued
  martingales (resp. local martingales, semimartingales). Given two
  $\mathbb{C}$-valued local martingales $M^1,M^2$ then
  $\overline{M^1},\overline{M^2}$ are still local martingales. Moreover
  $\langle \overline{M^1},\overline{M^2}\rangle=\overline{\langle
    M^1,M^2 \rangle}\ .$
 If $M$ is a $\C$-valued martingale then $\langle M, \overline{M} \rangle$
is a real valued increasing process.

All the local martingales admit a cadlag version. By default, 
when we speak about local martingales we always refer to their
cadlag version.
Given a real cadlag stochastic process $X$, the quantity 
$\Delta X_t$ will represent the jump $X_t - X_{t-}$. 
More details about previous notions are given in chapter I  of~\cite{JS03}.
%
%

For any special semimartingale X we define
$|| X||^2_{\delta^2}=\mathbb{E}\left[[M,M]_T\right]+\mathbb{E}\left(||A||_T^2\right)\ .
$
The set $\delta^2$ is the set of $(\mathcal{F}_t)$-special semimartingale $X$ for which $|| X||^2_{\delta^2}$ is finite.




\subsection{F\"ollmer-Schweizer Structure Condition}
\label{sec:FMStruc}

Let $X=(X_t)_{t \in [0,T]}$ be a real-valued
 special semimartingale with canonical decomposition, 
$
X= M+A.
$
For simplicity, we will  just suppose in the sequel that $M$
is a square integrable martingale.
For the clarity of the reader, we formulate in dimension one,
 the concepts appearing in the literature, see e.g. \cite{S94}
 in the multidimensional case.
%
For a given local martingale $M$, the space $L^2(M)$ consists of all predictable $\mathbb{R}$-valued processes $v=(v_t)_{t \in [0,T]}$ such that 
$
\mathbb{E}\left[\int_0^T|v_s|^2d\left\langle M\right\rangle_s\right]<\infty,
$
where $\langle M \rangle := \langle M, M \rangle$.
For a given predictable bounded variation process $A$, the space $L^2(A)$ consists of all predictable $\mathbb{R}$-valued processes $v=(v_t)_{t \in [0,T]}$ such that 
$
\mathbb{E}\left[(\int_0^T|v_s|d||A||_s)^2\right]<\infty\ .
$
Finally, we set
\begin{equation} \label{Theta}
\Theta:=L^2(M)\cap L^2(A), 
\end{equation}
which will be the class of {\bf admissible strategies}.
For any $v \in \Theta$, the stochastic integral process
$
G_t(v):=\int_0^tv_sdX_s,\quad \textrm{for all}\ t \in [0,T]\ ,
$
is therefore well-defined and is a semimartingale in $\delta^2$.
We can view this stochastic integral process as the gain process
associated 
with strategy $v$ on the underlying process $X$. 

The \textbf{minimization problem} we aim to study is the following. 
Given $H \in \mathcal{L}^2(\Omega,\shf,\P)$, 
a pair $(V_0,\varphi)$, where $V_0 \in \R$ and $\varphi \in \Theta$
is called \textbf{optimal} if $(c,v)=(V_0,\varphi)$ minimizes the 
expected squared hedging error
\begin{equation}\mathbb{E}[(H-c-G_T(v))^2]\ ,\label{problem2}\end{equation}
over all 
 pairs $(c,v) \in \mathbb{R} \times \Theta$.
 $V_0$ will represent the {\bf initial capital} of the hedging portfolio 
for the contingent claim $H$ at time zero. The definition below
introduces an important technical condition, see \cite{S94}.

\begin{defi}\label{defSC}
Let $X=(X_t)_{t \in [0,T]}$ be a real-valued special semimartingale. 
$X$ is said to satisfy the
 \textbf{structure condition (SC)} if there is a predictable
 $\mathbb{R}$-valued process $\alpha=(\alpha_t)_{t \in [0,T]}$ such that 
the following properties are verified.
\begin{enumerate}
 \item  
 $A_t=\int_0^t \alpha_s d\left\langle M\right\rangle_s\
 ,\quad \textrm{for all}\ t \in [0,T];$ in particular 
$dA$ is absolutely continuous with respect to $d \langle M \rangle$,
in symbols we denote    $dA\ll d\left\langle M\right\rangle $. 
\item  ${\displaystyle \int_0^T \alpha^2_s d\left\langle 
M\right\rangle_s<\infty\ ,\quad P-}$a.s.
\end{enumerate}

\end{defi}
%
From now on, we will denote by $K=(K_t)_{t \in [0,T]}$ the cadlag process 
$K_t=\int_0^t \alpha^2_s d\left\langle M\right\rangle_s\ ,\ \textrm{for all}\ t\in [0,T]\ .$
This process will be called the \textbf{mean-variance trade-off} (MVT) 
 process.
Lemma 2 of \cite{S94} states the following.
\begin{propo}\label{lemmaTheta}
If $X$ satisfies (SC) such that $K_T$ is a bounded r.v.,
 then $\Theta=L^2(M)$.
\end{propo}

The structure condition (SC) appears  naturally
in applications to financial mathematics. In fact, it is mildly related
to the no arbitrage condition at least when $X$ is a continuous process. 
Indeed, in the case where $X$ is  a continuous  martingale under 
 an equivalent probability measure, then 
 (SC) is fulfilled.
 

\subsection{F\"ollmer-Schweizer Decomposition 
and variance optimal hedging}
\label{sec:FMDecomp}

Throughout this section, as in Section~\ref{sec:FMStruc}, $X$ is supposed to be an $(\mathcal{F}_t)$-special semimartingale fulfilling the (SC) condition.
Two $(\mathcal{F}_t)$-martingales $M,N$ are said to be \textbf{strongly 
orthogonal} if $MN$ is a uniformly integrable martingale, see
 Chapter~IV.3 p.~179 of~\cite{Pr92}. 
If $M,N$ are two square integrable martingales,
then $M$ and $N$ are strongly orthogonal if and only if
 $\left\langle M,N\right\rangle=0$. This can be proved using
 Lemma~IV.3.2 of~\cite{Pr92}.

\begin{defi}\label{DefFSDecomp}
A random variable 
$H \in \mathcal{L}^2(\Omega,\mathcal{F},\P)$ admits a \textbf{F\"ollmer-Schweizer (FS) decomposition}, if 
\begin{eqnarray}H=H_0+\int_0^T\xi_s^HdX_s+L_T^H\ ,
 \quad P-a.s.\label{FSdecompo}\ ,\end{eqnarray}
where $H_0\in \mathbb{R}$ is a constant, $\xi^H \in \Theta$ and $L^H=(L^H_t)_{t \in [0,T]}$ is a square integrable martingale,  
with $\mathbb{E}[L_0^H]=0$ and  strongly orthogonal to $M$.
\end{defi}
We summarize now some fundamental results stated in
 Theorems 3.4 and 4.6, of~\cite{MS95} on the existence and uniqueness of the FS decomposition and of solutions for the optimization problem~\eqref{problem2}. 
\begin{thm} \label{ThmExistenceFS} 
We suppose that $X$  satisfies (SC) and that the MVT process $K$
 is uniformly bounded in $t$ and $\omega$.
Let $H \in \mathcal{L}^2(\Omega,\mathcal{F},\P)$.
\begin{enumerate}
\item $H$ admits a FS decomposition.
 It is unique in the sense that $H_0 \in \R$, 
$\xi^H \in L^2(M)$ and $L^H$ is uniquely
 determined by $H$.
\item For every $H \in \mathcal{L}^2(\Omega,\mathcal{F}, \P)$ and every 
$c \in \mathcal{L}^2(\mathcal{F}_0)$, there exists a unique strategy
 $\varphi^{(c,H)}\in \Theta$ such that
\begin{eqnarray}\mathbb{E}[(H-c-G_T(\varphi^{(c,H)}))^2]=\min_{v\in \Theta}\mathbb{E}[(H-c-G_T(v))^2]\label{problem1}\ .\end{eqnarray}
\item For every $H \in \mathcal{L}^2(\Omega,\mathcal{F},\P)$
  there exists a unique couple
 $(c^{(H)},\varphi^{(H)}) \in \mathcal{L}^2(\mathcal{F}_0) \times
  \Theta$
 such that
\begin{eqnarray*}\mathbb{E}[(H-c^{(H)}-G_T(\varphi^{(H)}))^2]=
\min_{(c,v)\in \mathcal{L}^2(\mathcal{F}_0) \times \Theta}\mathbb{E}[(H-c-G_T(v))^2]\ .\end{eqnarray*}
\end{enumerate}
\end{thm}

Next theorem gives the explicit form of the optimal strategy 
under some restrictions on $K$.
\begin{thm}\label{ThmSolutionPb1}
Suppose that X satisfies (SC) and that the MVT process $K$ of X is
deterministic and let  $\alpha$ be the process appearing in 
Definition~\ref{defSC} of~(SC).  Let $H\in \mathcal{L}^2(\Omega,\shf,\P)$ 
with FS decomposition (\ref{FSdecompo}).
\begin{enumerate}
\item 
For any $c \in \R$, the solution of the  minimization problem~(\ref{problem1})
 verifies  
$\varphi^{(c,H)} \in \Theta$, such that
\begin{eqnarray}\label{xisolution}
\varphi^{(c,H)}_t=\xi^H_t+\frac{\alpha_t}{1+\Delta K_t}(H_{t-}-c-G_{t-}
(\varphi^{(c,H)}))\ ,\quad\textrm{for \ all}\ t \in [0,T]
\end{eqnarray}
where the process $(H_t)_{t \in [0,T]}$ is defined by
$H_t:=H_0+\int_0^t\xi_s^HdX_s+L_t^H$.
 \item  The solution of the minimization problem~(\ref{problem2}) is
 given by the pair $(H_0,\varphi^{(H_0,H)})\ .$
\item If $\left \langle M \right\rangle$ is continuous, 
\begin{eqnarray*}
\min_{v\in \Theta}\mathbb{E}[(H-c-G_T(v))^2]&=&\exp(-K_T)\left((H_0-c)^2+\mathbb{E}[(L_0^H)^2]\right)\\
&&+\mathbb{E}\left[\int_0^T\exp\{-(K_T-K_s)\}d\left\langle L^H\right\rangle_s\right]\ .\label{corro9Sc2}\end{eqnarray*}
\end{enumerate}
\end{thm}
\begin{proof} \
Item 1. is stated in Theorem 3 of \cite{S94}. 
Item 2. is a consequence of Corollary 10 of \cite{S94}.
Item 3. is a consequence of Corollary 9 of 
\cite{S94} taking into account that $K $ 
inherits the continuity property
of  $\langle M \rangle$. We remark that 
$\tilde K = K$, 
where  $\tilde K$ is a process appearing in the statement
of the mentioned corollary.
\end{proof}

In the sequel, we will find an explicit expression of the
 FS decomposition for a large class of square integrable random
 variables, 
when the underlying process is an  exponential of additive process.

\subsection{Additive processes}

This subsection deals with     processes with independent
increments which are continuous in probability.
From now on $(\shf_t)$ will always be the canonical filtration
associated with $X$.

\begin{defi}\label{defPAI}
A cadlag process $X=(X_t)_{t \in [0,T]}$ is a (real) 
\textbf{additive process} iff
 $X_0=0$,
$X$ is continuous in probability, i.e. 
 $X$ has no fixed time of discontinuities and it has independent
 increments in the following sense:
 $X_t-X_s$ is independent of $\mathcal{F}_s$ for $0 \leq s < t \leq
  T$. 
 \\  $X$ is called {\bf L\'evy process}  if it is additive and 
the distribution of $X_t-X_s$ only depends on $t-s$ 
 for $0 \leq s \leq t \leq T$.
\end{defi}
 An important notion, in the theory of semimartingales,
 is the notion of characteristics, introduced in definition II.2.6 
of~\cite{JS03}.  A triplet of  \textbf{characteristics} 
 $(b,c,\nu)$, depends on a fixed truncation function  
 $h: \mathbb{R} \rightarrow \mathbb{R}$ with
 compact support such that $h(x)=x$ in a neighborhood of $0$;
 $\nu$ is some random $\sigma$-finite Borel measure on
$[0,T] \times \R$.  
If $X$ is a semimartingale additive process 
the triplet $(b,c,\nu)$ admits a deterministic version,
 see  Theorem II.4.15 of~\cite{JS03}. 
Moreover $(b_t)$, $(c_t)$ and 
 $t\mapsto \int_{[0,t] \times B} (\vert x \vert^2 \wedge 1) \nu(ds,dx)  $ 
 have bounded variation for any Borel real subset $B$.
Generally in this paper  $\mathcal{B}(E)$ denotes the Borel
 $\sigma$-field associated
with a topological space $E$.

\begin{propo}\label{propoJSA}
Suppose $X$ is a semimartingale additive process 
 with
characteristics $(b,c,\nu)$, where $\nu$ is
a non-negative Borel measure on $[0,T] \times \R$.
Then  $t\mapsto a_t$ given by 
\begin{eqnarray}a_t=||b||_t+c_t+\int_{\mathbb{R} }
 (|x|^2\wedge 1)\nu([0,t],dx) \label{171}\end{eqnarray}
fulfills
\begin{eqnarray}
db_t\ll da_t\ ,\quad
dc_t\ll da_t \quad \textrm{and}\quad 
\nu(dt,dx)= {F}_t(dx)da_t\ ,
\label{170}\end{eqnarray}
where $ {F}_t(dx)$ is a non-negative kernel from $\big([0,T],
\mathcal{B}([0,T])\big )$ into $(\mathbb{R} ,\mathcal{B}(\R) )$
verifying
\begin{equation} \label{E1000}
\int_\R  (\vert x \vert^2 \wedge 1)  F_t(dx) \le 1\,, \quad \forall t
\in [0,T].
\end{equation}
\end{propo}
\begin{proof}
The existence of $(a_t)$ as a process fulfilling $(\ref{171})$ 
and $ F$ fulfilling \eqref{E1000} 
is provided by the statement and the proof of  Proposition II.
2.9 of~\cite{JS03}. 
$(\ref{171})$
guarantees that $(a_t)$ is deterministic.
\end{proof}

We come back to the cumulant generating function $\kappa$ and its domain $D$.
\begin{remarque} \label{CumGenPAI}
In the case where the underlying process $X$ is an additive process, then 
\begin{eqnarray*}
 D := \{z \in \mathbb{C}\ \vert \ \mathbb{E}[e^{Re(z)X_t}]<\infty, \  \forall t \in
 [0,T]\} =\{z \in \mathbb{C}\ \vert \ \mathbb{E}[e^{Re(z)X_T}]<\infty\}\ .
\end{eqnarray*}
In fact, for given $t \in [0,T], \gamma \in \mathbb{R} $ we have
$   \mathbb{E}(e^{\gamma X_T}) = \mathbb{E}(e^{\gamma X_t}) 
\mathbb{E}(e^{\gamma (X_T - X_t)})  < \infty.$
Since each factor is positive,   if the left-hand side is finite,
then $\mathbb{E}(e^{\gamma X_t})$ is also 
 finite.
 \end{remarque}

\section{ F\"ollmer-Schweizer decomposition for exponential of additive processes}
\label{sec:expPII}

\setcounter{equation}{0}

%
 The aim of this section is to derive a quasi-explicit formula of the FS decomposition for exponential of additive processes with possibly non stationary increments. \\
%
 We assume that the process $S$ is the 
discounted price of the non-dividend paying stock which is supposed to be of the form,
 $S_t=s_0\exp(X_t),$ for all $ t\in[0,T],$
where $s_0$ is a strictly positive constant and $X$ is a
semimartingale additive process,
 in the sense of Definition~\ref{defPAI},  but not necessarily
 with stationary increments.
In the whole paper, if $z$ is a complex number,
$S^z_t$ stands for $\exp ({\rm ln} (s_0)  + z X_t )$.
In particular if  $y$ is a real number, 
$S^y_t$ stands  for $s_0 \exp(y X_t)$.

\medskip

\subsection{On some properties of cumulant generating functions}

We need now a result which extends the classical 
L\'evy-Khinchine decomposition, 
see e.g. 2.1 in Chapter II and Theorem 4.15 of Chapter II,
\cite{JS03}, which is only defined in the imaginary axis to
the whole domain of the cumulant generating function.
Similarly to Theorem~25.17 of~\cite{Sa99}, applicable
for the L\'evy case, 
 for additive processes we have the following.
\begin{propo}\label{LKdecoKappa}
Let $X$ be a semimartingale additive process and set
 $D_0=\left\{c \in \mathbb{R}\ |\ 
\int_{[0,T]\times \left\{|x|>1\right\}}e^{cx}\nu(dt,dx)<\infty \right\}$. Then, 
\begin{enumerate}
\item $D_0$ is convex and contains the origin.
\item $D_0= D\cap\mathbb{R}$.
\item If $z\in \mathbb{C}$ such that $Re(z)\in D_0$, i.e. $z\in D$, then
\begin{eqnarray} \label{E3.5bis}
\kappa_t(z)=z b_t+\frac{z^2}{2}c_t+\int_{[0,t]
\times\mathbb{R}} (e^{zx}-1-z h(x))\nu(ds,dx)\ .\end{eqnarray}
\end{enumerate}
\end{propo}

\begin{proof}
\begin{enumerate}
\item is a consequence of H\"older inequality similarly as i) in
 Theorem 25.17 of~\cite{Sa99}\ .
\item 
 The characteristic function of the law of $X_t$ is given through
the characteristics of $X$, i.e.
$$
\Psi_t(u)=iu b_t-\frac{u^2}{2}c_t+\int_{\mathbb{R}} 
(e^{iu x}-1-iu h(x))F^t(dx)\ ,\quad \textrm{for all}\ u\in\mathbb R,
$$ 
where we recall that for any $t \ge 0$, $c_t \ge 0$
and   $B\mapsto F^t(B) := \nu([0,t]\times B)$ 
is a positive measure which integrates $1 \wedge |x|^2$. 
Let $t \in [0,T]$.
 According to Theorem~II.8.1~(iii) of \cite{Sa99}, there is an
infinitely divisible distribution with characteristics 
$(b_t,c_t, F^t(dx))$. 
By uniqueness of the characteristic function, that law is precisely the
law of $X_t$. By Corollary~II.11.6, in~\cite{Sa99}, there is a L\'evy process
$(L^t_s, 0 \le s \le 1)$ such that $L^t_1$ and $X_t$ are identically 
distributed. We define
$$
C_0^t = \{ c\in \mathbb{R} \ \vert \ 
\int_{\{\vert x \vert > 1\}} e^{cx} F_t(dx) < \infty \} \quad\textrm{and}\quad
C^t = \{z \in \mathbb C \ \vert\  \mathbb{E}\left[\exp (Re (z L_1^t)\right ] < \infty \}  \ .
$$
Remark~\ref{CumGenPAI} says that $C^T = D$, moreover clearly
$C_0^T = D_0$.
Theorem V.25.17 of~\cite{Sa99} implies $D_0 = D \cap \mathbb{R}$,
i.e. point 2. is established.
\item Let $t \in [0,T]$ be fixed; let $z \in D \subset  C^t$, in
particular
$Re(z) \in C^t_0$.
We apply point (iii) of Theorem V.25.17 of~\cite{Sa99} 
to the L\'evy process $L^t$.
\end{enumerate}
\end{proof}

\begin{propo}\label{corroR1}
Let $X$ be a semimartingale additive process.
For all $z \in D$, 
 $t\mapsto \kappa_t(z)$  has bounded variation and
$\kappa_{dt}(z)\ll da_t,$ where $t \mapsto a_t$ was defined in Proposition
\ref{propoJSA}.
\end{propo}

\begin{proof}

Using (\ref{E3.5bis}), we only have to prove that
$t\mapsto \int_{[0,T]\times \R}(e^{zx}-1-zh(x))\nu(ds,dx)$
is absolutely continuous w.r.t.  $(da_t)$. 
We can conclude
 \begin{eqnarray*}\kappa_t(z) = 
 z \int_0^t \frac{db_s}{da_s} da_s +\frac{z^2}{2}
 \int_0^t \frac{dc_s}{da_s} da_s  \nonumber +
 \int_0^tda_s\int_{\mathbb{R}}\left(e^{zx}-1-zh(x)\right) {F}_s(dx)\ ,
\end{eqnarray*}
if we show that
\begin{eqnarray}\int_0^T da_s\int_{\R}|e^{zx}-1-zh(x)| {F}_s(dx)<\infty\ . \label{A131}\end{eqnarray}
Without restriction of generality we can suppose 
$h(x)=x1_{|x|\leq 1}$.
(\ref{A131}) can be bounded by the sum $I_1+I_2+I_3$ where
%
$$
I_1=\int_0^Tda_s\int_{|x|>1}|e^{zx}| {F}_s(dx)\ ,\quad 
I_2=\int_0^Tda_s\int_{|x|>1} {F}_s(dx)\ ,\quad\textrm{and}\quad 
I_3=\int_0^Tda_s\int_{|x|\leq1}|e^{zx}-1-zx| {F}_s(dx)\ .
$$
Using Proposition~\ref{propoJSA}, we have
$$
I_1= \int_0^Tda_s\int_{|x|>1}|e^{zx}| {F}_s(dx)
= \int_0^Tda_s\int_{|x|>1}e^{Re(z)x} {F}_s(dx) 
= \int_{[0,T]\times \{ |x|>1\} }e^{Re(z)x}\nu(ds,dx);
$$
this quantity is finite because $Re(z)\in D_0$ taking into account
 Proposition~\ref{LKdecoKappa}. Concerning $I_2$ we have
$$
I_2 = \int_0^Tda_s\int_{|x|>1} {F}_s(dx) 
= \int_0^Tda_s\int_{|x|>1}(1\wedge |x^2|) {F}_s(dx) 
\leq a_T,
$$
because of \eqref{E1000}. As far as $I_3$ is concerned, we have
$$
I_3 \leq e^{Re(z)}\frac{\vert z\vert^2}{2}\int_{[0,T]\times \{|x|\leq
  1\} }da_s(x^2\wedge 1) {F}_s(dx)= e^{Re(z)}\frac{\vert z \vert^2}{2}a_T
$$
again because of \eqref{E1000}. This concludes the proof of the proposition.

\end{proof}

The converse of the first part of previous Proposition \ref{corroR1} also holds.
To show this, we formulate first a simple remark.

\begin{remarque}\label{R3}
\begin{enumerate}
\item For every $z \in D$, $\left(\exp(zX_t-\kappa_t(z))\right)$
 is a martingale.
In fact, for  all $0\leq s\leq t\leq T$, we have
$\mathbb{E}[\exp(z(X_t-X_s))]=\exp(\kappa_t(z)-\kappa_s(z)).$
\item    $t \mapsto \kappa_t(0) \equiv 1$
and it has always bounded variation.
\end{enumerate}
\end{remarque}

\begin{propo} \label{P312}
Let $X$ be an additive process and $z \in D \cap \mathbb{R^\star}$.
$(X_t)_{t \in [0,T]}$ is a semimartingale if and only if
$t \mapsto \kappa_t(z)$ has bounded variation.
\end{propo}
\begin{proof} 
Using Proposition~\ref{corroR1}, 
it remains to prove the converse implication.
If $t \mapsto  \kappa_t(z)$ has  bounded variation then 
$ t \mapsto e^{\kappa_t(z)}$
 has the same property. Remark~\ref{R3} says that
 $e^{z X_t} = M_t e^{\kappa_t (z)}$ where $(M_t)$ is a martingale.
 Finally, $(e^{z X_t})$ is a semimartingale  
and taking the logarithm 
 $(zX_t)$  has the same property.
\end{proof}

\begin{remarque}\label{R10}
Let $z \in D$. If $(X_t)$ is a semimartingale additive process,
  then $(e^{zX_t})$
 is necessarily a special semimartingale
since it is the product of a martingale and a
 bounded variation continuous deterministic function and by use
of integration by parts.
\end{remarque}

\begin{propo}\label{proposition38}
The function $(t,z)\mapsto \kappa_t(z)$ is continuous. In
particular, $(t,z) \mapsto \kappa_t(z)$, $t \in [0,T]$, 
$z$ belonging to a compact real subset, is bounded.
\end{propo}
\begin{proof}
\begin{itemize}
\item Proposition~\ref{LKdecoKappa} implies that 
$z \mapsto \kappa_t(z)$ is continuous uniformly w.r.t.  $t \in [0,T]$.
\item We first prove  that 
$z \in {\rm Int}(D)$, $t\mapsto \kappa_t(z)$ is continuous. 
Since $z \in {\rm Int}(D)$, there is $\gamma > 1$ such that $\gamma z \in D$; so 
\begin{eqnarray*}\mathbb{E}[\exp(z\gamma X_t)]=\exp(\kappa_t(\gamma z))\leq \exp(\sup_{t\leq T}(\kappa_t(\gamma z))) \ ,\end{eqnarray*}
because $t\mapsto \kappa_t(\gamma z)$ is bounded, being of bounded
variation. This implies that $(\exp(zX_t))_{t \in [0,T]}$ is uniformly
integrable. Since $(X_t)$ is continuous in probability, then
$(\exp(zX_t))$ is continuous in $\mathcal{L}^1$. The partial 
result easily
follows.
\item To conclude it remains to show that  $t\mapsto \kappa_t(z)$ is
 continuous
for every $z \in D$.
Since $\bar D = \overline {\rm Int}(D)$, there is a sequence $(z_n)$ in the interior of $D$ 
converging to $z$. Since a uniform limit of continuous functions on $[0,T]$ 
is a continuous function, the result follows.
\end{itemize}
\end{proof}

\subsection{A reference variance measure} \label{SSRVM}

For notational convenience we introduce the set
$\frac{D}{2} = \{ z \in \C \vert 2z \in D\} $.
\begin{remarque} \label{Rdivise2}
We recall that   $D$ is convex. Consequently we have.
\begin{enumerate}
\item If $y, z \in \frac{D}{2}$, then $y + z \in D$.
If $z \in \frac{D}{2}  $ then $\bar z \in \frac{D}{2} $
and $2 Re(z) \in D$.
\item Since  $0 \in D$, clearly
$\frac{D}{2} \subset D$.
\item Under Assumption~\ref{HypD} below, $2 \in D$
and so $\frac{D}{2} + 1 \subset D$.
\end{enumerate}
\end{remarque}

%
We introduce a new function that will be useful in the sequel. 
\begin{defi}\label{defi:rho} 
\begin{itemize}
\item For any $t\in [0,T]$, if  $z,y \in  \frac{D}{2} $ 
we denote
\begin{equation}
\label{eq:rho:zy}
\rho_t(z,y)=\kappa_t(z+y)-\kappa_t(z)-\kappa_t(y)\ .
\end{equation}
\item To shorten notations $\rho_t:  \frac{D}{2}
 \rightarrow \C$
 will  denote the real valued function  such that, 
\begin{equation}
\label{L1}
\rho_t(z)= \rho_t(z,\bar z)=\kappa_t(2Re(z))-2Re(\kappa_t(z))\ . 
\end{equation}
Notice that the latter equality results from Remark~\ref{remarkR2} 1. 
\end{itemize}
\end{defi}
%
An important technical lemma follows below.
\begin{lemme}\label{L2}
Let $z\in \frac{D}{2}$, with $Re(z) \neq 0$, then, 
$t\mapsto \rho_t(z)$ is strictly increasing if and only if $X$ has no deterministic increments.
\end{lemme}
\begin{proof}
It is enough to show that $X$ has no deterministic increments if and only if for any $0\leq s<t\leq T$, the following quantity is positive, 
\begin{equation}
\label{eq1}
\rho_t(z)-\rho_s(z)=\big [\kappa_t\big(2Re(z)\big)-\kappa_s\big(2Re(z)\big)\big ]-2Re\big (\kappa_t(z)-\kappa_s(z)\big ) \ .
\end{equation}
By Remark~\ref{R3}, 
we have
$\exp[\kappa_t(z)-\kappa_s(z)]=\mathbb{E}[\exp(z\Delta_s^tX )],$
where  $\Delta_s^tX:=X_t-X_s.
$
Applying this property and Remark~\ref{remarkR2} 1.,
 to the exponential of the first term on the right-hand side of \eqref{eq1} yields
\begin{equation*}
\exp\left [\kappa_t\big(2Re(z)\big)-\kappa_s\big(2Re(z)\big)\right ]
=\mathbb{E}[\exp(2Re(z)\Delta_s^tX )]
= 
\mathbb{E}[\exp((z+\bar z)\Delta_s^tX )]
=
\mathbb{E}[\left|\exp(z\Delta_s^tX )\right|^2]\ .
\end{equation*}
Similarly, for the exponential of the second term on the right-hand 
side of
\eqref{eq1}, one gets
\begin{eqnarray*}
\exp\left [2Re\big (\kappa_t(z)-\kappa_s(z)\big )\right ]
&=&
\exp\left [\big(\kappa_t(z)-\kappa_s(z)\big)+\overline{\big(\kappa_t(z)-\kappa_s(z)\big)}\right]
=
\left|\mathbb{E}[\exp(z\Delta_s^tX )]\right|^2\ .
\end{eqnarray*}
Hence taking the exponential of $\Delta_s^t\rho(z):=\rho_t(z)-\rho_s(z)$ yields
\begin{eqnarray}
\label{E33}
\exp[\Delta_s^t\rho(z)]-1
&=&\frac{\mathbb{E}[\left|\exp(z\Delta_s^tX
    )\right|^2]}{\left|\mathbb{E}[\exp(z\Delta_s^t X)]\right|^2}-1 
=\frac{\mathbb{E}[\left|\Gamma_s^tX(z)\right|^2]}{\left|\mathbb{E}[\Gamma_s^tX(z)]\right|^2}-1\ ,\quad \textrm{where} \ \Gamma_s^tX(z)=\exp(z\Delta_s^tX )\ ,
\nonumber \\ &&\\
&=&\frac{Var\left [Re\big (\Gamma_s^tX(z)\big )\right ]+Var\left [Im\big
    (\Gamma_s^tX(z)\big )\right
  ]}{\left|\mathbb{E}[\Gamma_s^tX(z)]\right|^2}. \nonumber
\end{eqnarray}
\begin{itemize}
\item If $X$ has a deterministic increment $\Delta_s^tX =X_t-X_s$, then $\Gamma_s^tX(z)$ is again deterministic and~ \eqref{E33} vanishes and hence
  $t \mapsto \rho_t(z)$ is not strictly increasing.
\item If $X$ has never deterministic increments, then the nominator is never zero,
 otherwise $Re\big (\Gamma_s^tX(z)\big) = \exp( Re(z) \Delta_s^tX)$,
and therefore $\Delta_s^tX$ 
would be deterministic.
\end{itemize}
\end{proof}

\begin{remarque}\label{33bis}
 If $2\in D$, setting $z=1$ in~(\ref{E33}) implies  
that $\rho_t(1)>\rho_s(1)$ is equivalent to ${\displaystyle 
\frac{Var\big(\exp(\Delta_s^tX )\big)}{\big(\mathbb{E}[\exp(\Delta_s^tX ) ]
\big)^2}}>0$. 
Taking the process $S$ at discrete instants $t_0=0<\cdots <t_k<\cdots<t_n=T$, 
 one can define the discrete time process $(S^d_k)_{k=0,\cdots, n}$ such that  $S^d_k=S_{t_k}$ and derive the counterpart of Lemma~\ref{L2} in the discrete time setting. 
Indeed, the following assertions are equivalent: 
\begin{itemize}
	\item  $(\rho_{t_k}(1))_{k=0,\cdots n}$ is an increasing sequence;
	\item $\Delta_{t_k}^{t_{k+1}} X$ is never deterministic
 for any $k=0,\cdots, n-1$. 
	\end{itemize}
Moreover, accordingly to Proposition~3.10 in~\cite{GORDis}, we observe that,
 under one of the above equivalent conditions, the 
(discrete time) \textit{mean-variance trade-off process} associated with  $(S^d_k)_{k=0,\cdots, n}$ defined by 
$$
K^d_n:={\displaystyle 
\sum_{k=0}^{n-1}\frac{\big (\mathbb{E}[S_{k+1}-S_{k}  \vert \shf_k]\big)^2}
{\big( Var [S_{k+1}-S_{k} \vert \shf_k ]
\big)^2}}={\displaystyle \sum_{k=0}^{n-1}
\frac{\big (\mathbb{E}[\exp( \Delta_{t_k}^{t_{k+1}} X)-1\big)^2}
{\big( Var \left[ \exp(\Delta_{t_k}^{t_{k+1}} X )\right ]
\big)^2}}
$$ 
is always bounded. According to Proposition 2.6 of \cite{S95bis}, that condition
guarantees that every square integrable random variable admits a discrete F\"ollmer-Schweizer
decomposition.
 The process $K^d$ is the discrete analogous
of the MVT  process $K$;
 one can compare the mentioned result to  item 1. of Theorem \ref{ThmExistenceFS}.
\end{remarque}
%
From now on, we will always suppose the following assumption.
\begin{Hyp}\label{HypD}
\begin{enumerate}
\item $(X_t)$ has no deterministic increments.
\item $2 \in D$. 
\end{enumerate}
\end{Hyp}

We continue with a simple observation.
\begin{lemme}\label{lemme38a}
Let $I$ be a compact real interval included in $D$.
Then $\sup_{x \in I}
\sup_{t\leq T}\mathbb{E}[S_t^{x}]<\infty.$
\end{lemme}
\begin{proof}
Let $t \in [0,T]$ and $x\in I$,  since $\kappa$ is continuous, we have \\
$
\mathbb{E}[S_t^{x}]=
s_0^{x}\exp\{\kappa_t(x)\}\leq 
\max(1,s_0^{\sup I})\exp(\sup_{t\leq T,x \in I} \vert \kappa_t(x)
\vert)\ .
$
\end{proof}
%
%

\begin{remarque}\label{remarque313}
From now on, in this section, $d\rho_t=\rho_{dt}$ will denote the measure 
\begin{equation}
\label{eq:rho}
 d\rho_t=\rho_{dt}(1)=d(\kappa_t(2)-2\kappa_t(1))\ .
\end{equation}
According to
Assumption \ref{HypD} and Lemma \ref{L2},  it is a positive measure which is strictly positive on each interval.
This measure will play a fundamental role. \\
We state below a result that will help us to show that $\kappa_{dt}(z)$ is absolutely continuous w.r.t.
  $\rho_{dt}(1)$.
\end{remarque}

\begin{lemme}\label{propoRusso}
We consider two positive finite non-atomic Borel measures on $E \subset \mathbb{R}^n$, $\mu$ and $\nu$. We suppose the following:
\begin{enumerate}
\item $\mu \ll  \nu\ ;$
\item $\mu(I)\neq 0$ for every open ball $I$ of $E$.
\end{enumerate} 
Then $h := {\displaystyle \frac{d\mu}{d\nu} \neq 0}$ $\nu$ a.e. In particular
 $\mu$ and $\nu$ are equivalent.
\end{lemme}
\begin{proof}
We consider the Borel set
$
B=\{x\in E|h(x)=0\}.$
We want to prove that $\nu(B)=0$. So we suppose that there exists a
constant $c>0$ such that $\nu(B)=c>0$
and take another constant $\epsilon$ such that $0<\epsilon<c$. 
Since $\nu$ is a Radon measure, there are compact subsets 
 $K_{\epsilon}$ and $K_{\frac{\epsilon}{2}}$ of $E$ such that
$
K_{\epsilon} \subset K_{\frac{\epsilon}{2}} \subset B$ and 
$ \nu(B - K_{\epsilon})<\epsilon,\quad 
\nu(B - K_{\frac{\epsilon}{2}})<\frac{\epsilon}{2}.$
Setting $\epsilon=\frac{c}{2}$, we have
$
\nu(K_{\epsilon})>\frac{c}{2}$
and $ \nu(K_{\frac{\epsilon}{2}})>\frac{3c}{4}.
$
By Urysohn lemma, there is a continuous function 
$\varphi:E\rightarrow \mathbb{R}$ such that, $0 \leq \varphi \leq 1 $ with
$\varphi=1$ on $  K_{\epsilon}$ 
and $ \varphi=0 $
on the closure of $  K_{\frac{\epsilon}{2}}^c.$
Now
$
\int_E \varphi(x)\nu(dx) \ge \nu(K_{\epsilon})>\frac{c}{2}>0.
$
By continuity of $\varphi$ there is an open set $O \subset E$ with
$\varphi(x)>0$ for $x \in O$. Clearly $O \subset K_{\frac{\epsilon}{2}}
\subset B$; since $O$ is relatively compact, it is a countable union of
balls, and so $B$ contains a ball $I$. The fact that $h=0$ on $I$
implies $\mu(I)=0$ and this 
contradicts Hypothesis 2. of the statement.
Hence the result follows.

\end{proof}

\begin{remarque}\label{remarque3123}
\begin{enumerate}
\item If $E=[0,T]$, then point 2. of Lemma~\ref{propoRusso} becomes $\mu(I)\neq 0$ for every open interval $I\subset [0,T]$.
\item The result holds for every normal metric locally connected space
 $E$, provided  $\nu$ are Radon measures.
\end{enumerate} 
\end{remarque}

\begin{propo}\label{remarque314Bis}
Under Assumption~\ref{HypD} 
\begin{equation}
 d(\kappa_t(z))\ll  d\rho_t\ , \quad \textrm{for all} \ 
z \in D\label{AncHyp1}\ .
\end{equation}
\end{propo}
\begin{proof}
We apply Lemma~\ref{propoRusso}, with $d\mu=d\rho_t$ and $d\nu=da_t$. 
 Indeed, Proposition~\ref{corroR1} implies Condition 1. of Lemma~\ref{propoRusso} and Lemma~\ref{L2} implies Condition 2. of Lemma~\ref{propoRusso}. Therefore, $da_t$ 
is equivalent to $d\rho_t$.
\end{proof}
\begin{remarque}\label{R314bis}
Notice that this result also holds  with $d\rho_t(y)$ instead of 
$d\rho_t=d\rho_t(1)$, for any $y \in \frac{D}{2}
$ such that $Re(y) \neq 0$. 
\end{remarque}

\subsection{On some semimartingale decompositions and
 covariations} \label{OSSDAC}
%

\begin{propo}\label{corro36}
We suppose the validity of item 2. of  Assumption~\ref{HypD}. 
Let $y,z \in  \frac{D}{2}$.
 Then $S^z$ is a special semimartingale whose canonical decomposition
 $S_t^z =
M(z)_t+A(z)_t$ satisfies
\begin{equation}\label{LL2}
A(z)_t=\int_0^tS_{u-}^z\kappa_{du}(z)\ ,
\quad 
\left\langle M(y),M(z)\right\rangle_t= 
\int_0^tS_{u-}^{y+z}\rho_{du}(z,y)\ , \quad
M(z)_0 = s_0^z,
\end{equation}
where  $d\rho_u(z)$ is defined by equation~(\ref{L1}). 
In particular we have the following: 
\begin{enumerate} 
\item \label{MAexpPII} 
$ \left\langle M(z),M\right\rangle_t=\int_0^tS_{u-}^{z+1}\rho_{du}(z,1)$
\item $\left\langle M(z),M(\bar z)
\right\rangle_t=\int_0^tS_{u-}^{2Re(z)}\rho_{du}(z)\ .$ \label{L22}
\end{enumerate}

\end{propo}
\begin{remarque}\label{remarque312B}
\begin{itemize}
\item Clearly $1 \in D$ since  $0$ and $2$  belong to $D_0$
and $D_0$ is convex by Proposition \ref{LKdecoKappa}. 
\item
If $z = 1$, we have $S^z = S$, so that by uniqueness
of the  special semimartingale decomposition,
it follows that
$M(1) = M$. 
\end{itemize}

\end{remarque}

\begin{proof}
The case $y = 1$, follows very similarly to the proof of
Lemma 3.2 of \cite{Ka06}. The major tools are integration by parts
and  Remark~\ref{R3} which says that $N(z)_t:=e^{-\kappa_t(z)}S_t^z$
 is a martingale. The general case can be easily adapted.
\end{proof}

\begin{remarque}\label{remarque312a}
Lemma~\ref{lemme38a} implies that $\mathbb{E}\left[|\left\langle M(y),M(z)
\right\rangle\right|]<\infty$ and so $M(z)$ is a square integrable 
martingale for any $z \in \frac{D}{2}$.
\end{remarque}


\subsection{On the Structure Condition}
Proposition~\ref{corro36} with $y=z=1$ yields 
$S = M+A$ where 
$
A_t=\int_0^tS_{u-}\kappa_{du}(1)$
and $M$ is a martingale such that 
$\left\langle M,M\right\rangle_t=\int_0^tS_{u-}^{2}
(\kappa_{du}(2)-2\kappa_{du}(1))=\int_0^tS_{u-}^{2}\rho_{du}.
$
At this point, the aim is to exhibit a predictable $\mathbb{R}$-valued 
process $\alpha$ such that 
\begin{enumerate}
 \item $A_t=\int_0^t \alpha_s d\left\langle M\right\rangle_s, t \in [0,T]$.
 \item $K_T=\int_0^T \alpha^2_s d\left\langle M\right\rangle_s$ 
is bounded. 
\end{enumerate}
In that case, according to item 1. of 
Theorem~\ref{ThmExistenceFS}, there will exist a unique FS decomposition
for any $H \in \mathcal{L}^2(\Omega,\shf,\P)$ and so the minimization
problem~(\ref{problem2}) will have a unique solution, characterized by
 Theorem~\ref{ThmSolutionPb1} 2.

\begin{propo}\label{propoabsconti} Under Assumption~\ref{HypD}, 
$A_t=\int_0^t \alpha_s d\left\langle M\right\rangle_s,$
where $\alpha$ is given by
\begin{eqnarray} \label{alphau}
\alpha_u:=\frac{\lambda_u}{S_{u-}} \quad \textrm{with}\quad \lambda_u:=\frac{d\kappa_{u}(1)}{d\rho_u}\ ,\quad \textrm{for all}\ u\in [0,T].\end{eqnarray}
Moreover the MVT process is given by 
\begin{eqnarray}K_t=\int_0^t \left(\frac{d(\kappa_{u}(1))}{d\rho_u}\right)^2d\rho_u\ .\label{KPII}\end{eqnarray}
\end{propo}

\begin{corro}  \label{C320}
Under Assumption~\ref{HypD},
the structure condition (SC) is verified if and only if 
\begin{eqnarray*}K_T=\int_0^T \left(\frac{d(\kappa_{u}(1))}{d\rho_u}\right)^2d\rho_u<\infty\ .\end{eqnarray*}
In particular, $(K_t)$ is deterministic therefore bounded.
\end{corro}
\begin{remarque} \label{RNDI}
Item 1. of Assumption~\ref{HypD} is natural.
Indeed if it were not
 realized, i.e.
if $X$ admits a deterministic increment on some interval 
$[s,t]$,  then $S$ would not fulfill 
the (SC) condition, unless $u\mapsto \kappa_u(1)$ is
constant on $[s,t]$. In this case, the market model would admit 
arbitrage  opportunities.
\end{remarque}
\begin{prooff} \ (of Proposition \ref{propoabsconti}).
By Proposition~\ref{remarque314Bis},  $d\kappa_{t}(1)$ is absolutely continuous w.r.t.  $d\rho_t$. Setting $\alpha_u$ as in~(\ref{alphau}), relation
~(\ref{KPII}) follows from Proposition~\ref{corro36}, expressing $K_t=\int_0^t\alpha_u^2d\left\langle M\right\rangle_u$. 
\end{prooff}

\begin{lemme}\label{lemmeTheta}
The space $\Theta$, defined in \eqref{Theta}, is constituted by all predictable processes $v$ such that
$\mathbb{E}\left(\int_0^Tv_t^2S_{t-}^{2}d\rho_t
\right) <\infty.$
\end{lemme}
\begin{proof}
According to Proposition~\ref{lemmaTheta}, the fact that $K$ is bounded and $S$ satisfies (SC), then $v \in \Theta$ holds if and only if $v$ is predictable and $\mathbb{E}[\int_0^Tv^2_td\left\langle M,M\right\rangle_t]<\infty$. Since
$\left\langle M,M\right\rangle_t=\int_0^tS_{s-}^{2}d\rho_s, $
the assertion follows.
\end{proof}

\subsection{Explicit F\"ollmer-Schweizer decomposition}

%
We denote by $\shd$ the set of $z \in D$ such that
\begin{equation} \label{Eqshd}
   \int_0^T\left \vert \frac{d\kappa_u(z)}{d\rho_u}\right \vert^2
d\rho_u<\infty.
\end{equation}
From now on, we formulate another assumption.
%
\begin{Hyp}\label{Hyp1bis}
$1 \in \shd$.
\end{Hyp}

\begin{remarque} \label{R323} 
\begin{enumerate}
\item 
Because of Proposition~\ref{remarque314Bis}, ${\displaystyle 
\frac{d\kappa_t(z)}{d\rho_t}}$ exists for every $z \in D$.
\item Under  Assumption~\ref{HypD}, Corollary \ref{C320} says that
 Assumption~\ref{Hyp1bis} is equivalent to (SC).
\end{enumerate}
\end{remarque}
The proposition below will constitute an important step for determining
the FS decomposition of the contingent claim $H=f(S_T)$
for a significant class of functions $f$, see 
Section~\ref{SSDSCC}.
\begin{propo}\label{lemme38} 
Let $z \in \shd \cap \frac{D}{2}$
with $ z+1  \in \shd$, 
(in particular 
$2 Re(z)    \in D$), then 

\begin{enumerate}
\item 
 $S_T^z \in \mathcal{L}^2(\Omega,\mathcal{F}, \P)$.
\item  Moreover, suppose that Assumptions~\ref{HypD} and~\ref{Hyp1bis} hold and define  
\begin{equation} \label{gammaZT}
 \gamma(z,t) :=  \frac{d(\rho_{t}(z,1))}{d\rho_t}, \ t \in [0,T]. 
\end{equation}
Then 
 $ \int_0^T \vert \gamma(z,t)\vert^2 \rho_{dt} < \infty$ 
and 
\begin{equation} \label{etaZT}
\eta(z,t) := \kappa_{t}(z)-\int_0^t\gamma(z,s)\kappa_{ds}(1) 
= \kappa_{t}(z)-\int_0^t\gamma(z,s) \frac{d\kappa_s(1)}{d\rho_{s}}
 \rho_{ds}   
\end{equation}
is well-defined, besides  
 $\eta(z,\cdot)$ is absolutely
continuous w.r.t.  $\rho_{ds}$ and therefore bounded.
\item Again under  Assumptions~\ref{HypD} and~\ref{Hyp1bis},
 $H(z)=S_T^z$ admits an FS decomposition 
$H(z)=H(z)_0+\int_0^T\xi(z)_tdS_t+L(z)_T$ where
\begin{eqnarray}
H(z)_t &:= &e^{\int_t^T\eta(z,ds)}S_t^z\ , \label{FS1}  \\
\xi(z)_t&:=&\gamma(z,t)e^{\int_t^T\eta(z,ds)}S_{t-}^{z-1}\ , \label{FS2} \\
L(z)_t&:=&H(z)_t-H(z)_0-\int_0^t\xi(z)_udS_u\ . \label{FS3}
\end{eqnarray}

\end{enumerate}
\end{propo}

\begin{proof}
\begin{enumerate}
\item is a consequence of Lemma~\ref{lemme38a}.
\item $\gamma (z, \cdot)$ is square integrable because Assumption
\ref{Hyp1bis} and $z, z+1 \in \shd$. 
Moreover $\eta$ is well-defined  since
\begin{equation} \label{FS10}
 \left(\int_0^T \vert \gamma(z,s) \vert
\left \vert \frac{d\kappa_s(1)}{d\rho_{s}} \right \vert \rho_{ds} \right)^2 \le 
 \int_0^T \vert \gamma(z,s)\vert^2 \rho_{ds}
\int_0^T \left \vert \frac{d\kappa_s(1)}{d\rho_{s}} \right \vert^2 \rho_{ds} .
\end{equation}
\item
In order to prove that~(\ref{FS1}),(\ref{FS2}) and~(\ref{FS3})
is
 the FS decomposition of $H(z)$,
 we need to show that
\begin{enumerate}
\item $H(z)_0$ is $\mathcal{F}_0$-measurable,
\item $\left\langle L(z),M\right\rangle=0,$
\item $\xi(z)\in\Theta,$ where $\Theta$ was defined in \eqref{Theta}.
\item $L(z)$ is a square integrable martingale.
\end{enumerate}
We proceed similarly to the proof of Lemma 3.3 of \cite{Ka06}.
Point (a) is obvious. 
Partial integration and point~\ref{MAexpPII} of Proposition~\ref{corro36}
yield
\begin{equation} \label{FHZ}
H(z)_t
= H(z)_0+\int_0^te^{\int_u^T\eta(z,ds)}dM(z)_{u}-\int_0^t 
e^{\int_u^T\eta(z,ds)}S_u^z \eta(z,du)
+\int_0^te^{\int_u^T\eta(z,ds)}S_{u-}^z\kappa_{du}(z)  \ .
\end{equation}
On the other hand
\begin{equation} \label{FAXI}
\int_0^t\xi(z)_udS_u 
= \int_0^t\xi(z)_udM_u+\int_0^t\gamma(z,u)
e^{\int_u^T\eta(z,ds)}S_{u-}^{z}\kappa_{du}(1)\ .
\end{equation}
Hence, using expressions \eqref{FHZ} and \eqref{FAXI}, 
 by definition of $\eta$ in (\ref{etaZT}), which says
$
\eta(z,du)=\kappa_{du}(z)-\gamma(z,u)\kappa_{du}(1),
$ we obtain
\begin{equation} \label{43}
L(z)_t = H(z)_t-H(z)_0-\int_0^t\xi(z)_udS_u =
\int_0^te^{\int_u^T\eta(z,ds)}dM(z)_{u}-\int_0^t\xi(z)_udM_u,
\end{equation}
which implies that $L(z)$ is a local martingale.\\
From point~\ref{MAexpPII}.  of Proposition~\ref{corro36},
using \eqref{FS2},  
it follows that
%
%
$$ \left\langle
  L(z),M\right\rangle_t=\int_0^te^{\int_u^T\eta(z,ds)}S_{u-}^{z+1}
[\rho_{du}(z,1)-\gamma(z,u)\rho_{du}]. $$
Then by definition of $\gamma$ in (\ref{gammaZT}), 
$\rho_{dt}(z,1)=\gamma(z,t)\rho_{dt}\ ,$
 yields $
\left\langle L(z),M\right\rangle_t=0.$
Consequently, point (b) follows.

It remains to prove point (d) i.e. that $L(z)$ is a square-integrable martingale
 for all $z \in D$ and that $Re(\xi(z))$ and $Im(\xi(z))$ are in $\Theta$. 
(\ref{43}) says that
\begin{eqnarray*}
L(z)_t=\int_0^te^{\int_s^T\eta(z,du)}dM_s(z)-\int_0^t\xi(z)_sdM_s\ .
\end{eqnarray*}
By Remark~\ref{remarkR2} we observe first that 
 $\bar z, \bar z + 1 \in \shd$. Moreover by
 definition of $\gamma$
and $\eta$, it follows  
\begin{equation} \label{lemmepassagecomplex}
 \overline{\gamma(z,t)}=\gamma(\bar{z},t) \quad {\rm and} 
\quad \overline{\eta(z,t)}=\eta(\bar{z},t).
\end{equation}
By Proposition~\ref{corro36}, \ref{lemmepassagecomplex} and~\eqref{43},  it follows
\begin{eqnarray} \label{44ter}
\left\langle L(z),\overline{L(z)}\right\rangle_t
&=&\left\langle L(z),L(\bar{z})\right\rangle_t 
= \left\langle
L(z),\int_0^{.}e^{\int_s^T\eta(\bar{z},du)}dM_s(\bar{z})
\right\rangle_t \nonumber \\
%
&=&\int_0^te^{\int_s^T\eta(z,du)}e^{\int_s^T\eta(\bar{z},du)}S_{s-}^{2Re(z)}
\rho_{ds}(z)
-\int_0^t\xi(z)_se^{\int_s^T\eta(\bar{z},du)}S_{s-}^{1+\bar{z}}\rho_{ds}(\bar{z},1)
 \ .
\end{eqnarray}
Consequently
\begin{equation} \label{Eellz}
\left\langle L(z),\overline{L(z)}\right\rangle_t=
\int_0^te^{\int_s^T 2Re(\eta(z,du))}S_{s-}^{2Re(z)}
[\rho_{ds}(z) -|\gamma(z,s)|^2\rho_{ds}]\ .
\end{equation}
Taking the expectation in~(\ref{Eellz}), using  point 2.,
 \eqref{gammaZT}, \eqref{etaZT} and
Lemma~\ref{lemme38a}, we obtain
\begin{eqnarray}\mathbb{E}\left[\left\langle L(z),\overline{L(z)}\right\rangle_T\right]<\infty\label{LsquareMart}\ .\end{eqnarray}
Therefore, $L$ is a square-integrable martingale. \\
It remains to prove point (c) i.e. that  $\xi(z) \in \Theta$.
In view of applying Lemma~\ref{lemmeTheta}, we evaluate
\begin{equation} \label{XiXiz}
\int_0^T|\xi(z)_s|^2S_{s-}^2 \rho_{ds}
= \int_0^T|\gamma(z,s)|^2e^{\int_t^T2Re(\eta(z,du))}S_{s-}^{2Re(z)}\rho_{ds}. 
\end{equation}
 Similarly as for (\ref{Eellz}), we can show that  the expectation 
of the right-hand side of (\ref{XiXiz}) is finite. 
This concludes the proof of Proposition~\ref{lemme38}.
\end{enumerate}
\end{proof}



\subsection{FS decomposition of special contingent claims} 
\label{SSDSCC}

 We consider now payoff functions of the type 
\begin{eqnarray}
H=f(S_T) \quad\textrm{with}\quad f(s)=\int_{\mathbb{C}} s^z\Pi(dz)\ ,
\label{FORM}\end{eqnarray}
where $\Pi$ is a (finite) complex measure in the sense of Rudin~\cite{RU87},  Section~6.1. 
An integral representation of some basic European calls is provided in the sequel.
We need now the new following assumption. 
\begin{Hyp}\label{Hyp1}
Let $I_0= {\rm supp} \Pi\cap \mathbb{R}$. 
We denote $I = 2I_0 \cup \{1 \}.$
\begin{enumerate}
\item $I_0$ is compact.
\item $\forall z \in {\rm supp} \Pi, \quad z, z+1 \in \shd.$
\item $I_0 \subset \frac{D}{2}$.
\item  $\sup_{x \in I  }\left\|\frac{d(\kappa_t(x))}{d\rho_t}
\right\|_{\infty}<\infty$.
\end{enumerate}
\end{Hyp}
\begin{remarque}\label{R325}
\begin{enumerate}
\item Two kinds of assumptions appear.  Assumptions \ref{HypD}
and \ref{Hyp1bis} only concern the process and Assumption
\ref{Hyp1} involves both the process and the payoff.
\item
Assumption \ref{Hyp1} looks obscure. Examples for its validity will
be provided in Section~\ref{SecExamples}. For instance
consider  the specific case where $X$ is a Wiener integral driven by a
L\'evy process $\Lambda$, i.e. 
$X_t = \int_0^t l(s) d \Lambda_s, t \in [0,T]$
and the payoffs are either a call or a put.
We observe in Example \ref{E517} below that 
Assumptions~\ref{HypD},~\ref{Hyp1bis} and~\ref{Hyp1} are a consequence
of the simple Assumption~\ref{HypWiener}.
\end{enumerate}
\end{remarque}

\begin{remarque}\label{remarque320}
\begin{enumerate}
\item Point 3. of Assumption~\ref{Hyp1} implies
 $\sup_{z \in I + i \R}
\left\|\kappa_{dt}(Re(z))\right\|_{T}<\infty\,.$
\item Under Assumption~\ref{Hyp1}, $H = f(S_T)$ is square integrable.
In particular it admits an FS decomposition.
\item Because of~(\ref{AncHyp1}) in Proposition~\ref{remarque314Bis}, 
the Radon-Nykodim derivative at Assumption~\ref{Hyp1}.4, always exists.
\end{enumerate}
\end{remarque}
We need now to obtain upper bounds on $z$ for the
quantity~(\ref{LsquareMart}). We will first need the following lemma 
which constitutes a (not straightforward) generalization 
of Lemma 3.4 of \cite{Ka06} which was stated when $X$  is a L\'evy process.
The fact that $X$ does not have stationary increments, constitutes 
 a significant obstacle.

%
\begin{lemme}\label{lemmeconstante}
Under  Assumptions~\ref{HypD},
  \ref{Hyp1bis},~\ref{Hyp1},
there are positive constants $c_1, c_2, c_3$ such that $d\rho_s$ a.e.
 \begin{enumerate}
\item
 $
{\displaystyle \sup_{z\in {I_0 + i\R}}
 \frac{dRe(\eta(z,s))}{d\rho_s} \le c_1\,. 
 }$
\item For any $z \in I_0 + i \R \ ,$
$\quad 
{\displaystyle
 \vert \gamma(z,s) \vert^2 \le    \frac{d\rho_s(z)}{d\rho_s} \le
c_2 - c_3
\frac{dRe(\eta(z,s))}{d\rho_s} 
}\,.$ 
\item $\ 
{\displaystyle - \sup_{z \in {I_0 + i \R}} \int_0^T 2 Re (\eta(z,dt)) \exp \left(
 \int_t^T 2 Re (\eta(z,ds)) \right) < \infty\,.
 }$
\end{enumerate}
\end{lemme}
\begin{remarque}\label{R328}
\begin{enumerate} 
\item According to Proposition~\ref{lemme38},
$t \mapsto Re (\eta(z,t))$ is absolutely continuous 
w.r.t.  $d\rho_t$.
\item We recall that ${\rm supp} \Pi$ is included in $I_0 + i \R$.
\end{enumerate}
\end{remarque} 
 \begin{prooff} \ (of {\bf  Lemma}~\ref{lemmeconstante}).
According to Point 3. of Assumption~\ref{Hyp1} 
we denote 
\begin{equation} \label{E300}
 c_{11} :=  
\sup_{x \in I}\left 
\Vert\frac{d(\kappa_t(x))}{d\rho_t}\right\Vert_{\infty}. 
\end{equation}
For $z \in  I_0 + i \R, t\in [0,T]$, we have
$
\eta(z,t)=\kappa_{t}(z)-\int_0^t\gamma(z,s)d\kappa_{s}(1)
\ \textrm{and}\ \eta(\bar{z},t)=
\kappa_{t}(\bar{z})-\int_0^t\gamma(\bar{z},s)d\kappa_{s}(1)\,.
$
Then, we get
$Re(\eta(z,t))=Re(\kappa_{t}(z))-\int_0^tRe(\gamma(z,s))d\kappa_{s}(1)\ .$
We obtain
\begin{eqnarray}\label{D0}
  \int_t^TRe(\eta(z,ds)) &\leq&  Re\left(\kappa_T(z)-\kappa_t(z)\right)
+ \left \vert \int_t^T \gamma(z,s) d\kappa_s(1) \right \vert  \nonumber\\
&=&  \int_t^T \frac{Re (d\kappa_{s}(z))}{d\rho_s} d\rho_s
+ \left \vert \int_t^T \gamma(z,s) d\kappa_s(1) \right \vert \ .
\end{eqnarray}
 Since $\left\langle L(z),\overline{L(z)}\right\rangle_t$ is increasing, 
 taking into account \eqref{Eellz}, 
 the measure  $\left (d\rho_s(z)-|\gamma(z,s)|^2d\rho_s\right)$ is non-negative. It follows that
\begin{eqnarray} \label{D0bis}
\frac{d\rho_s(z)}{d\rho_s}-|\gamma(z,s)|^2\geq 0\ ,\quad d\rho_s \ a.e.
\label{MajorationGamma}\end{eqnarray}
%
By  \eqref{MajorationGamma},
in particular  the density ${\displaystyle \frac{d\rho_s(z)}{d\rho_s}}$
 is non-negative $d\rho_s$ a.e.
 Consequently,  
\begin{eqnarray} 
\label{H02}
 2\frac{dRe(\kappa_s(z))}{d\rho_s}\leq 
\frac{d\kappa_s(2Re(z))}{d\rho_s}\ ,
\quad d\rho_s\ a.e. 
\end{eqnarray}
%
In order to prove 1. it is enough to verify that, for some $c_0 > 0$,
\begin{eqnarray}
\label{EB1}
\frac{dRe(\eta(z,s))}{d\rho_s}\leq c_0+
\frac{1}{2}\frac{dRe(\kappa_s(z))}{d\rho_s} & d\rho_s \ a.e.
 &
\end{eqnarray}
In fact, \eqref{H02},  Assumption \ref{Hyp1} point 3. and \eqref{E300},
 imply that
${\displaystyle 
\frac{dRe(\eta(z,s))}{d\rho_s}
\leq c_0 + \frac{1}{2} c_{11} =: c_1. 
}$
%
To prove \eqref{EB1} it is enough to show that 
\begin{equation} 
\label{D0ter}
Re(\eta(z,T)-\eta(z,t))
\leq c_{0}(\rho_T-\rho_t)+\frac{1}{2}Re(\kappa_T(z)-\kappa_t(z)),
 \quad \forall t \in [0,T]. 
 \end{equation}
Again Assumption~\ref{Hyp1} point 3. implies that
\begin{equation} 
\label{EA11}
\left |\int_t^T\gamma(z,s)d\kappa_s(1)\right| \le c_{12} 
\int_t^T \vert \gamma(z,s) \vert d\rho_s\ ,
 \end{equation}
where $c_{12} = \Vert \frac{d\kappa_s(1)}{d\rho_s} \Vert_\infty.$
Using \eqref{MajorationGamma} and  Assumption~\ref{Hyp1} it follows
%
\begin{equation} \label{EC1}
\left|\gamma(z,s)\right|^2 \leq \frac{d\rho_s(z)}{d\rho_s}=\frac{d\kappa(2Re(z))}{d\rho_s}-\frac{2dRe(\kappa_s(z))}{d\rho_s}
\leq  c_{11}-\frac{2dRe(\kappa_s(z))}{d\rho_s}.
  \end{equation}
This implies that
$ c_{12}^2 \left|\gamma(z,s)\right|^2 \le
 \left(c_{13}^2+\frac{1}{4}
\left (\frac{dRe(\kappa_s(z))}{d\rho_s}\right)^2 \right),
$
where $c_{13} > 0$ is chosen such that 
 $c_{13}^2 \ge 4 c_{12}^4 +  c_{12}^2 c_{11}$.
%
Consequently, 
\begin{eqnarray*}\left|\int_t^T\gamma(z,s)d\kappa_s(1)\right| \le 
\int_t^T d\rho_s \left(c_{13}+\frac{1}{2}
\left|\frac{dRe(\kappa_s(z))}{d\rho_s}\right|\right)\ .
\end{eqnarray*}
Coming back to \eqref{D0}, we obtain
\begin{eqnarray*}
Re(\eta(z,T)-\eta(z,t)) &\le &
\int_t^T \left ( \frac{Re(d\kappa_s(z))}{d\rho_s}  
   + c_{13} + \halb 
\left \vert \frac{Re(d\kappa_s(z))}{d\rho_s} \right \vert \right) d\rho_s  \\
&\le &  
\int_t^T \left ( \halb \frac{Re(d\kappa_s(z))}{d\rho_s}  
 + \left(\frac{Re(d\kappa_s(z))}{d\rho_s} \right)^+ 
 + c_{13} \right)  d\rho_s.
\end{eqnarray*}
\eqref{H02} and Assumption~\ref{Hyp1} allow to establish
\begin{eqnarray}Re(\eta(z,T)-\eta(z,t))\leq \int_t^T d\rho_s
\left(c_{0}+\frac{1}{2}\frac{dRe(\kappa_s(z))}{d\rho_s}\right)\ ,
\label{EA33}\end{eqnarray}
where $c_0=\frac{c_{11}}{2}+c_{13}$. This concludes the proof of point 1.\\
In order to prove point 2. we first observe that \eqref{EB1} implies
\begin{eqnarray}-\frac{dRe(\kappa_s(z))}{d\rho_s}\leq
 2\left(c_0-\frac{dRe(\eta(z,s))}{d\rho_s}\right)\quad d\rho_s\  a.e. 
 \label{EB3}\end{eqnarray}
\eqref{EC1} implies
%
$
{\displaystyle 
\left|\gamma(z,s)\right|^2\leq c_{21}-4\frac{dRe(\eta(z,s))}{d\rho_s}\ ,
}$
%
where $c_{21}=c_{11}+4c_0$. Point 2. is now established with $c_2=c_{21}$ and $c_3=4$.\\
We continue with the proof of point 3. We decompose
$Re(\eta(z,t))=A^+(z,t)-A^-(z,t),$
where $A^+(z,.)$ and $A^-(z,.)$ are the increasing non negative functions given by  \\
$
A^+(z,t)= \displaystyle \int_0^t\left(\frac{dRe(\eta(z,s))}{d\rho_s}\right)_+d\rho_s
 \ and \ 
A^-(z,t)= \displaystyle \int_0^t\left(\frac{dRe(\eta(z,s))}{d\rho_s}\right)_-d\rho_s.
$
\\
%
Moreover 
point 1. implies
$A^+(z,t)\leq c_1\rho_t.$
At this point, for $z \in  I_0 + i \R$ 
\begin{eqnarray*}
-\int_0^TRe(\eta(z,dt))e^{\int_t^T2Re(\eta(z,ds))}&=&
\int_0^T\left(A^-(z,dt)- A^+(z,dt)\right)e^{
2\int_t^TRe(\eta(z,ds))}\,\\
&\leq&\int_0^TA^-(z,dt)e^{2\left(A^+(z,T)-A^+(z,t)\right)}e^{-2\left(A^-(z,T)-A^-(z,t)\right)}\,\\
&\leq&e^{2c_1\rho_T}\int_0^Te^{-2\left(A^-(z,T)-A^-(z,t)\right)}A^-(z,dt)\,\\
&=&\frac{e^{2c_1\rho_T}}{2}\left\{1-e^{-2A^-(z,T)}\right\} 
\leq\frac{e^{2c_1\rho_T}}{2}\ ,
\end{eqnarray*}
which concludes the proof of point 3 of Lemma~\ref{lemmeconstante}.

%
\end{prooff}


\begin{thm}\label{propo310} Let $\Pi$ be a finite complex-valued 
Borel  measure on $\C$.
Suppose Assumptions~\ref{HypD},
  \ref{Hyp1bis},~\ref{Hyp1}.
Any complex-valued contingent claim $H=f(S_T)$, where
 $f$ is of the form~(\ref{FORM}), 
and $H \in \shl^2 (\Omega,\shf, \P) $,
 admits a unique FS decomposition $H=H_0+\int_0^T\xi^H_tdS_t+L^H_T$ with the
 following
properties.
\begin{enumerate} 
\item 
 $H \in \shl^2(\Omega, \shf,\P)$ and
$$ H_t=\int H(z)_t\Pi(dz), \quad
       \xi_t^H=\int \xi(z)_t\Pi(dz), \quad
      L_t^H=\int L(z)_t\Pi(dz),
$$
where for $z \in {\rm supp} (\Pi)$,  $H(z),\xi(z)$ and $L(z)$ 
are the same as those introduced in 
Proposition~\ref{lemme38}  and we convene that they vanish if 
$z \notin   {\rm supp} (\Pi)$.
\item Previous decomposition is real-valued
if $f$ is real-valued.
\end{enumerate}
\end{thm}
\begin{remarque} \label{Rsimbis}
Taking $\Pi = \delta_{z_0}(dz), \ z_0 \in \C$,
Assumption~\ref{Hyp1} is equivalent to the assumptions of Proposition 
\ref{lemme38}.
\end{remarque}

\begin{prooff}\
(of {\bf  Theorem~\ref{propo310}}).
{a)} 
$f(S_T) \in \shl^2(\Omega, \shf, \P)$ since by Jensen's,
$E \left \vert \int_\C \Pi(dz) S_T^z \right \vert ^2
  \le \int_\C \vert \Pi\vert(dz)
 E \vert S_T^{2Re(z)}\vert  \vert \Pi\vert(\C)
 \le \sup_{x \in I_0} E( S_T^{2x}) \vert \Pi\vert(\C)^2,$
where $|\Pi|$ denotes the total variation of the finite measure $\Pi$.
Previous quantity is bounded because of Lemma~\ref{lemme38a}. \\
{b)}
We go on with the FS decomposition.
 We would like to prove first that $H$ and $L^H$ are well defined
 square-integrable processes and 
$E(\int_0^T \vert \xi^H_s\vert^2 d\langle M\rangle_s) < \infty$.\\
By Jensen's inequality, we have
\begin{eqnarray*}\mathbb{E} \left \vert \int_\mathbb{C}
    L(z)_t\Pi(dz)\right \vert^2 \leq 
\mathbb{E} \left ( \int_\mathbb{C}|\Pi|(dz) |L(z)_t|^2 \right )
\vert \Pi (\C) \vert 
=\int_\mathbb{C}|\Pi|(dz)\mathbb{E}[|L(z)_t|^2] \vert \Pi(\C) \vert.
\end{eqnarray*}
Similar calculations allow to show that
\begin{eqnarray*}\mathbb{E}[(\xi^H)_t^2]\leq 
\vert \Pi \vert(\C) \int_\mathbb{C} |\Pi|dz)\mathbb{E}[|\xi (z)_t|^2]\ 
\quad \textrm{and}\quad 
\mathbb{E}[(L_t^H)^2]\leq \vert \Pi(\C)\vert
\int_\mathbb{C} |\Pi|(dz)\mathbb{E}[|L(z)_t|^2]\ .
\end{eqnarray*}
We will show now that
\begin{itemize}
\item (A1): $\sup_{t\leq T,z \in {\rm supp} \Pi }\mathbb{E}[|H_t(z)|^2]<\infty\,;$
\item (A2): $ \int_\mathbb{C}|\Pi|(dz)\mathbb{E}[|L(z)_T|^2] < \infty;$
\item (A3): 
$ E\left( \int_0^T d\rho_t S_t^2
\int_\C  \vert \xi_t(z)\vert^2 \vert \Pi \vert(dz) \right)<\infty.$
\end{itemize}
(A1): Since $H(z)_t=e^{\int_t^T\eta(z,ds)}S_t^z$, we have
$|H(z)_t|^2=H(z)_t\overline{H(z)_t}=e^{\int_t^T2Re(\eta(z,ds))}
S_t^{2Re(z)}\ ,$
so
$$
  \mathbb{E}[|H(z)_t|^2] =  e^{\int_t^T 2Re(\eta(z,ds))}
\mathbb{E}[S_t^{2Re(z)}]\leq c_4 e^{\int_t^T 2Re(\eta(z,ds))}\ ,
$$
where $c_4$ is well defined by~\eqref{EspStBornee}, below, since by Lemma~\ref{lemme38a}, 
\begin{equation} \label{EspStBornee}
c_4:=\sup_{x \in I, s\leq T}\mathbb{E}\left[S^x_s\right]
<\infty\ .
\end{equation}
%
Lemma~\ref{lemmeconstante} implies (A1). Therefore $(H_t)$ is a
well-defined square-integrable process.
%
%
(A2): 
$\mathbb{E}[|L_t(z)|^2]\leq \mathbb{E}[|L_T(z)|^2]=\mathbb{E}[\left\langle L(z),\overline{L(z)}\right\rangle_T]\ ,$
where the first inequality is due to the fact that $|L_t(z)|^2$ is a 
submartingale. 
\begin{eqnarray*}\mathbb{E}\left[\left\langle L(z),
\overline{L(z)}\right\rangle_T\right]=
\mathbb{E}\left[\int_0^Te^{\int_s^T2Re(\eta(z,du))}S_{s-}^{2Re(z)}
\left[d\rho_s(z)-|\gamma(z,s)|^2d\rho_s\right]\right] \ .\end{eqnarray*}
By Fubini's theorem, Lemma~\ref{lemme38a} and \eqref{Eellz}, we have
\begin{eqnarray*}
\mathbb{E}\left[\left\langle L(z),\overline{L(z)}\right\rangle_T\right]
&=& \int_0^Te^{\int_s^T2Re(\eta(z,du))} \mathbb{E}[S_{s-}^{2Re(z)}]
\left[\frac{d\rho_s(z)}{d\rho_s} -|\gamma(z,s)|^2 \right]
d\rho_s  \\
&\le & c_4 \int_0^Te^{\int_s^T2Re(\eta(z,du))}
\left[\frac{d\rho_s(z)}{d\rho_s} \right]
d\rho_s.
\end{eqnarray*}
%
According to Lemma~\ref{lemmeconstante} point 2, previous expression is 
bounded by $c_4 I(z)$, where 
\begin{equation}\label{EIZ}
I(z):=\int_0^Td\rho_t\exp\left(\int_t^T2Re(\eta(z,ds))\right) \left[c_2-c_3
\frac{dRe(\eta(z,t))}{d\rho_t}\right]
=  c_2I_1(z)+c_3I_2(z),
\end{equation}
where
$
I_1(z)=\int_0^Td\rho_t\exp\left(\int_t^T2Re(\eta(z,ds))\right)$ and
$I_2(z)= - \int_0^T\exp\left(\int_t^T2Re(\eta(z,ds))\right)Re(\eta(z,ds)).
$
Using again Lemma~\ref{lemmeconstante}, we obtain
\begin{equation} \label{EIZbis}
\sup_{z\in I_0+i\mathbb{R}}\left|I_1(z)\right|\leq
\rho_T\exp\left(2c_1\rho_T\right)\,\quad \textrm{and}\quad 
\sup_{z\in I_0+i\mathbb{R}}\left|I_2(z)\right|<\infty \ ,
\end{equation}
and so
\begin{eqnarray}
\sup_{z\in I_0+i\mathbb{R}} \mathbb{E}\left[\left\langle L(z),
\overline{L(z)}\right\rangle_T\right]
<\infty\ .
\label{EIZ1}\end{eqnarray}
This concludes (A2).\\
We verify now the validity of (A3). This requires to control
\begin{eqnarray*}\mathbb{E}\left[\int_0^T\rho_{dt}S_t^2\left(\int_{\mathbb{C}}|\Pi|(dz)|\xi(z)_t|^2\right)\right]\leq \mathbb{E}\left[\int_0^T\rho_{dt}S_t^2\left(\int_{\mathbb{C}}|\Pi|(dz)\left|\gamma(z,t)\exp\left(\int_t^TRe(\eta(z,ds))\right)S_t^{z-1}\right|^2\right)\right]\ .\end{eqnarray*}
Using Jensen's inequality, this is smaller or equal than  
\begin{eqnarray*}
\vert \Pi(\C) \vert \int_{\mathbb{C}}|\Pi|(dz)\int_0^T\rho_{dt}\mathbb{E}\left[S_t^{2Re(z)}\right]|\gamma(z,t)|^2\exp\left(2\int_t^TRe(\eta(z,ds))\right)\ .
\end{eqnarray*}
Lemma~\ref{lemmeconstante} gives the upper bound
$
c_4 \vert \Pi\vert(\C)  
\int_{\mathbb{C}}|\Pi|(dz)I(z),
$
where $I(z)$ was defined in \eqref{EIZ}.
 Since $\Pi$ is finite and because of \eqref{EIZbis}, (A3) is now
 established. \\
{c)}
In order to conclude the proof of item 1., it remains to show that $L$ is an
$(\shf_t)$-martingale which is strongly orthogonal to $M$.
This can be established similarly as in \cite{Ka06}, Proposition 3.1,
by making use of  Fubini's theorem and 
 Fubini's theorem for stochastic integrals (cf.~\cite{Pr92},
 Theorem IV.46) and (A1), (A2), (A3).\\

Consequently,  $(H_0,\xi^H,L^H)$ provide a (possibly complex)
 FS decomposition of $H$.
\item{d)}
It remains to prove item 2., that is to say 
that the decomposition is real-valued.
 Let $(H_0,\xi^H,L^H)$ and $(\overline{H_0},\overline{\xi}^H,\overline{L}^H)$
 be two FS decomposition of $H$. Consequently,
since $H$ and $(S_t)$ are real-valued, we have
$ 0=H-\overline{H}=(H_0-\overline{H}_0)+ \displaystyle 
\int_0^T(\xi^H_s-\overline{\xi}^H_s)dS_s+(L_T^H-\overline{L}^H_T),$
which implies that $0=Im(H_0)+\int_0^TIm(\xi^H_s)dS_s+Im(L_T^H)$.
 By Theorem~\ref{ThmExistenceFS} 1., the uniqueness of the real-valued
 F\"ollmer-Schweizer decomposition yields
 that the processes $(H_t)$,$(\xi_t^H)$ and $(L_t^H)$ are real-valued.
\end{prooff}

\subsection{Representation of call and put options}
\label{sec:payoff:laplace}

We used some integral representations of  payoffs of the
form~(\ref{FORM}). We refer to~\cite{Cramer},~\cite{Raible98} and more
recently~\cite{Eberlein}, for some characterizations of classes of
functions which admit this kind of representation. In order to apply
the results of this paper, we need explicit formulae
 for the complex measure $\Pi$ in some example of contingent claims. 
Let $K>0$ be a strike.
\begin{description}
\item{\bf The  European Call option $H=(S_T-K)_+$}. 
 For arbitrary $0<R<1$, $s>0$, we have
\begin{eqnarray}(s-K)_+-s=\frac{1}{2\pi i}\int_{R-i\infty}^{R+i\infty}s^z\frac{K^{1-z}}{z(z-1)}dz\ .\label{Call2}\end{eqnarray}
\item{\bf The European Put option $H=(K-S_T)_+$}.
%
For an arbitrary $R<0$, $s>0$, we have
\begin{eqnarray}(K-s)_+=\frac{1}{2\pi i}\int_{R-i\infty}^{R+i\infty}s^z\frac{K^{1-z}}{z(z-1)}dz\ .\label{Put1}\end{eqnarray}
\end{description}

\section{The solution to the minimization problem}

\setcounter{equation}{0}

FS decomposition will help to provide the solution to the global
minimization  problem.
Let $X$ be an additive process with cumulant generating function
$\kappa$.
We denote $S_t = s_0 \exp(X_t), t \in [0,T]$, $s_0 > 0$.
Next theorem deals with the case where the payoff to hedge is given as a
bilateral Laplace transform of the exponential of the additive process $X$.
 It is an
extension of Theorem~3.3 of~\cite{Ka06} to additive processes
 with no stationary increments. 
\begin{thm}\label{mainthm}
Let $H=f(S_T)$ where $f$ is of the form~(\ref{FORM}). 
We assume the validity of Assumptions~\ref{HypD},~\ref{Hyp1bis}, 
 \ref{Hyp1}.
The variance-optimal capital $V_0$ and the variance-optimal hedging strategy $\varphi$, solution of the minimization problem $(\ref{problem2})$, are given by
$V_0=H_0$
and the implicit expression
\begin{equation}\varphi_t=\xi^H_t+\frac{\lambda_t }{S_{t-}}(H_{t-}-V_0-\int_0^t\varphi_sdS_s)\ ,\end{equation}
where the processes $(H_t)$, $(\xi_t)$ and $(\lambda_t)$ are defined by
\begin{eqnarray*}\gamma(z,t)&:=&\frac{d\rho_{t}(z,1)}{d\rho_t} \quad \textrm{with}\quad \rho_t(z,y)=\kappa_t(z+y)-\kappa_t(z)-\kappa_t(y)\ ,\\
\eta(z,dt)&:=&\kappa_{dt}(z)-\gamma(z,t)\kappa_{dt}(1), \quad
\lambda_t :=\frac{d(\kappa_{t}(1))}{d\rho_t} \\
H_t&:=&\int_\C e^{\int_t^T\eta(z,ds)}S_t^z\Pi(dz), \quad
\xi^H_t:=\int_\C \gamma(z,t)e^{\int_t^T\eta(z,ds)}S_{t-}^{z-1}\Pi(dz)\ .
\end{eqnarray*}
The optimal initial capital is unique. The optimal hedging strategy 
$\varphi_t(\omega)$ is unique up to some $(\P(d\omega)\otimes dt)$-null set.
\end{thm}

\begin{remarque}\label{remarque328}
The mean variance trade-off process can be expressed as,
 see~(\ref{KPII}),
$K_t=\int_0^t\frac{d\kappa_u(1)}{d\rho_u}\kappa_{du}(1).$
\end{remarque}

\begin{prooff}\
(of {\bf Theorem~\ref{mainthm}}).

Since $K$ is deterministic, the optimality follows from
Theorem~\ref{propo310} and by items 1. and 2. of Theorem~\ref{ThmSolutionPb1}.
We recall that $\alpha$ was given in \eqref{alphau}.
 Uniqueness follows from Theorem
\ref{ThmExistenceFS} 2.
\end{prooff}

When the underlying price is an exponential of additive process, 
we evaluate the so called {\bf variance of the hedging error}
of the contingent claim $H$ i.e.
the  quantity  $\mathbb{E}[(V_0+G_T(\varphi)-H)^2]$,
where   $V_0,\varphi$ and $H$ were defined at  Theorem~\ref{mainthm}.
\begin{thm}\label{thm34}
Under the assumptions of Theorem~\ref{mainthm},
the variance of the hedging error equals
\begin{eqnarray*}J_0:=\left(\int_\mathbb{C}\int_\mathbb{C} J_0(y,z)\Pi(dy)\Pi(dz)\right)\ ,\end{eqnarray*}
where
$$
J_0(y,z):= \left \{
\begin{array}{ccc}
s_0^{y+z}\int_0^T\beta(y,z,t)e^{\kappa_{t}(y+z)+\alpha(y,z,t)} d\rho_t\ &:& y, z
 \in{\rm supp} \Pi \\
0 &:& otherwise,
\end{array}
\right.
$$
with
\begin{eqnarray}
\label{eq:Beta}
\alpha(y,z,t)&:=&\eta(z,T)-\eta(z,t)-(\eta(y,T)-\eta(y,t))-\int_t^T\left(\frac{d\kappa_{s}(1)}{d\rho_s}\right)^2d\rho_s\ ,\nonumber \\ 
\beta(y,z,t)&:=&\frac{d\rho_{t}(y,z)}{d\rho_t} -
\frac{d\rho_{t}(y,1)}{d\rho_t} \frac{d\rho_{t}(z,1)}{d\rho_t}
\ .
\end{eqnarray}
\end{thm}
%
This expression of the error involving the function
 $\beta$~\eqref{eq:Beta}, can be used to characterize the price
 models that are exponential of additive processes for which the
 market is complete, at least for vanilla option payoffs. 
For instance, by evaluating $\beta$, we can verify, in Remarks~\ref{remarquePoisson}
 and~\ref{remarqueGaussian}, below, the complete market model property in the Poisson and the Gaussian case. 
%
%
%
\begin{prooff} \
(of {\bf Theorem~\ref{thm34}}).
Since $X_0 = 0$, $\shf_0$ is the trivial $\sigma$-field, therefore $L^H_0 =
0$, because it is mean-zero and deterministic.

The quadratic error can be calculated using 
  Theorem~\ref{ThmSolutionPb1} 3. It gives
\begin{eqnarray}\mathbb{E}\left[\int_0^T\exp\left\{-(K_T-K_s)\right\}
d\left\langle L^H\right\rangle_s\right]\ ,\label{E1}\end{eqnarray}
where $L^H$ is the remainder martingale in the FS decomposition of $H$.
We proceed now to the evaluation of  
 $\left\langle L^H\right\rangle$.
Similarly to the proof of Theorem 3.2 of \cite{Ka06},
using 
(\ref{44ter}), 
the   bilinearity and the stability w.r.t.  complex conjugate
 of the covariation together with
\eqref{EIZ1},  
 it is possible
to show that
\begin{equation} \label{ECrochetL}
\left\langle L^H,L^H\right\rangle_t=\int\int\left\langle L(y),L(z)\right
\rangle_t\Pi(dy)\Pi(dz). 
\end{equation}
 It remains to evaluate
 $\left\langle L(y),L(z)\right\rangle$ for $y,z \in  {\rm supp (\Pi)}$.
We know by Proposition~\ref{corro36} that for all $y,z \in \frac{D}{2}$,
\begin{eqnarray*}\left\langle M(y),M(z)\right\rangle_t=\int_0^tS_{u-}^{y+z}\rho_{du}(y,z)\ .\end{eqnarray*}
Using the same terminology as in Proposition~\ref{lemme38},
similarly to \eqref{Eellz} we have
\begin{eqnarray*}
\left\langle L(y),L(z)\right\rangle_t
&=& \int_0^t e^{\int_s^T(\eta(z,du)+\eta(y,du))}S_{s-}^{y+z}
\left[\rho_{ds}(y,z)-\gamma(z,s)\rho_{ds}(y,1)\right] \\
&=&\int_0^t
  e^{\int_s^T(\eta(z,du)+\eta(y,du))}S_{s-}^{y+z}\beta(y,z,s) d\rho_s\ .
\end{eqnarray*}
We come back to~(\ref{E1}).
Recalling that 
$\alpha(y,z,t)=(\eta(z,T)-\eta(z,t))-(\eta(y,T)-\eta(y,t))-(K_T-K_t),$
where $K$ is the MVT process, 
 we have
$$
\int_0^T e^{-(K_T-K_t)}d\left\langle L(y),L(z)\right\rangle_t
= \int_0^T e^{\alpha(y,z,t)}S_{t-}^{y+z}\beta(y,z,t) d\rho_t.
$$
Since $\mathbb{E}[S_{t-}^{y+z}]=s_{0}^{y+z}e^{\kappa_t(y+z)}$, an application of Fubini's theorem yields
\begin{equation}\label{E2}
\mathbb{E} \left (\int_0^T e^{-(K_T-K_t)}d\left\langle
    L(y),L(z)\right\rangle_t
\right )
= s_{0}^{y+z}\int_0^T e^{\alpha(y,z,t)+\kappa_t(y+z)}\beta(y,z,t) d\rho_t,
\end{equation}
which equals $J_0(y,z)$.
\eqref{ECrochetL}, \eqref{E2} and again Fubini's theorem imply
$$
\int_0^Te^{-(K_T-K_t)}d\left\langle L^H,L^H\right\rangle_t=
\int_\mathbb{C}\int_\mathbb{C} J_0(y,z)\Pi(dy)\Pi(dz).
$$
This concludes the proof of Theorem~\ref{thm34}.
\end{prooff}
%




\section{Examples}

\label{SecExamples}


\subsection{Exponential of a Wiener integral driven by a L\'evy process}

\label{ExpWiener}


Let  $\Lambda$ be a square integrable
L\'evy process and let $(t,z) \mapsto \kappa_t^\Lambda(z)$
be the cumulative generating function of $\Lambda$ 
with domain $D_\Lambda$ in the sense of
Definition~\ref{defKappa}.  $(t,z) \mapsto \kappa_t^\Lambda(z)$
is continuous because of Proposition
\ref{lemme38}.
 We observe that
\begin{equation} \label{PropKappa}
\kappa_t^\Lambda(z) = t \kappa^\Lambda(z)\ ,
\end{equation}
where $\kappa^\Lambda: \Lambda \rightarrow \C$
is a continuous function such that  
$\kappa^\Lambda(z) =  \kappa^\Lambda_1(z).$
Let $l:[0,T] \rightarrow \R$ be a bounded Borel function. 
We will consider in this subsection  the additive process $X_t=\int_0^t l_s d\Lambda_s$.
Let us define the set $D_\Lambda(l)\subset \R$ such that 
$$
D_\Lambda(l) = \{ x \in \R \vert \underline{l} x, \overline{l} x \in D_\Lambda \} + i \R\ ,\quad \textrm{where}\quad  \underline{l} := \inf l, \quad \overline{l}:=  \sup l \ .
$$
\begin{lemme}\label{lemme27}
The cumulant generating function  of $X$
 is such that for all 
$ z \in D_\Lambda(l)$, we have
$$
\kappa_{X_t}(z)=\int_0^t\kappa^\Lambda(z l_s)ds.
$$
In particular $D_\Lambda(l) \subset D$, 
where $D$ is the domain defined according
to Definition~\ref{defKappa}.
\end{lemme}
\begin{proof} \
If  $l$ is
  continuous, the result follows from the observation
that 
 $\int_0^T l_s d\Lambda_s$ is the 
limit in probability of $\sum_{j=0}^{p-1} l_{t_j}(\Lambda_{t_{j+1}}-
\Lambda_{t_j})$ where $0=t_0<t_1<...<t_p=T$ is a subdivision of $[0,T]$
 whose mesh converges to zero. 
If $l$ is only Borel bounded the result can be established
through approximation by convolution.
\end{proof}

 We formulate the following hypothesis which will be in force for
the whole subsection.
\begin{Hyp}\label{HypWiener}
\begin{enumerate}
\item $\kappa^\Lambda(2) - 2\kappa^\Lambda(1) \neq 0.$
\item $\underline{l} > 0$ and $2 \overline{l} \in D_\Lambda$.
\end{enumerate}
\end{Hyp}
\begin{remarque}
\label{rem:kappa:gamma}
 Lemma~\ref{L2} applied to $X$ being the L\'evy process
$\Lambda$ implies that,  for every $\gamma>0$,  such that
 $2 \gamma \in D_\Lambda$, we have
\begin{eqnarray}\label{equa1}
\kappa^\Lambda(2\gamma)-2\kappa^\Lambda(\gamma)>0\ .
\end{eqnarray}
\end{remarque}
\begin{remarque} \label{R514bis}
\begin{enumerate}
\item By item 2. of Assumption \ref{HypWiener},
$2 \in  D_\Lambda(l)$ and so does
$1$ because $ D_\Lambda(l)$ is convex.
By  Lemma \ref{lemme27}, $1$ and $2$ belong to $D$.
\item $\rho_t=\int_0^t\left(\kappa^\Lambda(2 l_s)-
2\kappa^\Lambda(l_s)\right)ds\,;$
\item
$X$ is a   semimartingale additive process since 
$t \mapsto \kappa_t(2)$ has bounded variation, see Proposition~\ref{P312}.
\end{enumerate}
\end{remarque}

\begin{propo}\label{propo411}
Assumptions~\ref{HypD}  and~\ref{Hyp1bis} are verified.
Moreover $D_\Lambda(l) \subset \shd$.
\end{propo}
\begin{proof}
\begin{enumerate}
\item By item 1. of Remark \ref{R514bis}, $2 \in D$ and
so the second item of   Assumption~\ref{HypD} is verified.
Using Lemma~\ref{L2}, item 1. of Assumption~\ref{HypD} is verified if we
  show that $t\mapsto \rho_t(1)=\kappa_t(2)-2\kappa_t(1)$ is strictly
 increasing. Now
$\kappa_t(2)-2\kappa_t(1)=\int_0^t\left(\kappa^\Lambda(2l_s)-2\kappa^\Lambda(l_s)\right)ds.$
Inequality~(\ref{equa1}) and item 2. of Assumption \ref{HypWiener}  
 imply that $\forall s\in[0,T]$,
$\kappa^\Lambda(2l_s)-2\kappa^\Lambda(l_s)>0,$ 
and consequently $t \mapsto \rho_t(1)$ is strictly increasing.
\item 
For $z \in D_\Lambda(l)$, by Lemma \ref{lemme27} and Remark~\ref{R514bis} 2.
 we have
\begin{eqnarray} \label{F55}
\left|\frac{d\kappa_t(z)}{d\rho_t}\right|
=\left|\frac{\kappa^\Lambda(zl_t)}{\kappa^\Lambda(2l_t)-
2\kappa^\Lambda(l_t)}\right|\leq
 \frac{\sup_{x\in[\underline{l} ,\overline{l}]}\vert \kappa^\Lambda(xz) \vert}
{\inf_{x\in [\underline{l}, \overline{l}]}\left
(\kappa^\Lambda(2x)-2\kappa^\Lambda(x)\right)}\ .
\end{eqnarray}
Previous supremum and infimum exist since $x\mapsto \kappa^\Lambda(zx)$ is continuous and it attains a maximum and a minimum on a compact interval.
So, $D_\Lambda(l) \subset \shd$ and
 Assumption~\ref{Hyp1bis} is verified because of point 1. in
 Remark~\ref{R514bis}.
\end{enumerate}
\end{proof}

\begin{remarque}\label{remarque412}
Suppose for a moment that
\begin{equation} \label{E4000}
2 I_0 \subset \{x \vert \underline{l}  x, \overline{l} x \in D_\Lambda \}.
\end{equation}
\begin{enumerate}
\item That implies then $ 2 I_0 \subset D_\Lambda(l)$.
 Point 3. of Assumption  \ref{Hyp1}  follows by Lemma \ref{lemme27}.
Item  2. of the same Assumption  
is   also  verified. In fact, since 
$2 I_0  \subset D_\Lambda(l)$ and $2 \in  D_\Lambda(l)$
and because of the fact that  $D_\Lambda(l)$ is convex, we have
 $$ {\rm supp \Pi} \cup ({\rm supp \Pi}+1) \subset
 \frac{D_\Lambda(l)}{2} +  \frac{D_\Lambda(l)}{2} \subset 
D_{\Lambda}(l).$$
The conclusion follows by  Proposition \ref{propo411} which says
 $ D_\Lambda(l) \subset \shd$.
\item From the  proof of Proposition \ref{propo411}, it follows that
\begin{eqnarray*}\frac{d\kappa_t(z)}
{d\rho_t}=\frac{\kappa^\Lambda(zl_t)}
{\kappa^\Lambda(2l_t)-2\kappa^\Lambda(l_t)}\ .
\end{eqnarray*}
 Admitting point 1. of Assumption~\ref{Hyp1}, 
then $I$ is compact.
Taking into account \eqref{F55}, the fact that $1 \in D_\Lambda(l)$,
so $I \subset D_\Lambda(l)$, and that 
 $\kappa^\Lambda$ 
is continuous,
 point 4. of Assumption~\ref{Hyp1}
is verified.
\end{enumerate}
\end{remarque}
We consider again the same class 
of options as in previous subsections. 
To conclude the  verification of  Assumption~\ref{Hyp1}
it remains to show the following.
\begin{itemize}
\item $I_0$ is compact. This point will be trivially fulfilled
in the specific cases.
\item \eqref{E4000}.
\end{itemize}

%
\begin{example} \label{E517}
We keep in mind the call and put representations provided in
Section \ref{sec:payoff:laplace}.
\begin{enumerate}
\item $H=(S_T-K)_+$.
In this case $2 I_0  = \{2R,2 \}$ and \eqref{E4000} is verified,
since $R \in ]0,1[$.
\item $H=(K-S_T)_+$. Again, here $R<0$,
$2 I_0 = \{2R\}$. \\
Again, we only have to require that $D_\Lambda$
contains some negative values, which is the case for the three
examples introduced in Remark~\ref{remark47Bis}.
Selecting $R$  in a proper way, \eqref{E4000} is fulfilled.
\end{enumerate}
\end{example}

\begin{corro}\label{THMREPLEVY}
We consider a process $X$  of the form $X_t = \int_0^t l_s d\Lambda_s$
under Assumption~\ref{HypWiener}.
 The FS decomposition of an option $H$ of the type $(\ref{FORM})$ 
and the related solution to the minimization problem 
are provided by Theorem~\ref{propo310}, Proposition~\ref{lemme38}
and Theorem~\ref{mainthm} together with the expressions given below.\\
For $z \in {\rm supp} \Pi,  t \in [0,T]$ we have
\begin{eqnarray*}\lambda_s&=&\frac{\kappa^\Lambda(l_s)}{\kappa^\Lambda(2 l_s)
-2\kappa^\Lambda(ls)}, \quad
\gamma(z,s) = \frac{\kappa^\Lambda((z+1)l_s)-
\kappa^\Lambda(zl_s)-\kappa^\Lambda(l_s)}{\kappa^\Lambda(2l_s)-
2\kappa^\Lambda(l_s)}, \\
  \eta(z,s)&=&\kappa^\Lambda(zl_s)-\frac{\kappa^\Lambda(l_s)}
{\kappa^\Lambda(2l_s)-2\kappa^\Lambda(l_s)} 
\left(\kappa^\Lambda((z+1)l_s)-\kappa^\Lambda(zl_s)-\kappa^\Lambda(l_s)\right).\end{eqnarray*}
Again, for convenience, if $z \notin {\rm supp} \Pi$ then 
we define 
$ \gamma(z, \cdot ) = \eta(z,\cdot) \equiv 0.$
\end{corro}

\subsection{Considerations about the L\'evy case}

If $l \equiv 1$ then $X$ coincides with the L\'evy process $ \Lambda$ and    
Assumption~\ref{HypWiener} is equivalent to Hubalek et alia 
Condition introduced in \cite{Ka06} i.e.
1. \quad $2\in D\,;$ \quad
2. $\kappa^\Lambda(2)-2\kappa^\Lambda(1)\neq 0\,.$ \\

In that case we have $D = D_\Lambda = D_\Lambda(l)$.
Therefore $\shd = D $
because
$\displaystyle \frac{d\kappa_t}{d\rho_t}(z)=\frac{1}
{\kappa^\Lambda(2)-2\kappa^\Lambda(1)}\kappa^\Lambda(z)$ for any 
$t \in [0,T], z\in D.$\\

We recall some cumulant and log-characteristic functions 
of some typical L\'evy processes.
\begin{remarque}\label{remark47Bis}
\begin{enumerate}
\item \underline{Poisson Case:}
If $X$ is a Poisson process with intensity $\lambda$,
 we have that $\kappa^\Lambda(z)=\lambda(e^z-1)$.
 Moreover, in this case the set $D_\Lambda=\mathbb{C}$. \\
\item \underline{NIG Case:}
This process was introduced by Barndorff-Nielsen in~\cite{bnh}.
Then $X$ is a L\'evy process with 
$X_1 \sim NIG(\alpha,\beta,\delta,\mu)$, with $\alpha>|\beta|>0$, $\delta>0$
 and $\mu \in \mathbb{R}$.
 We have $\kappa^\Lambda(z)=\mu z + \delta (\gamma_0-\gamma_z)$ and
 $\gamma_z=\sqrt{\alpha^2-(\beta+z)^2}$, 
$D_\Lambda = [-\alpha-\beta,\alpha-\beta]+i\mathbb{R}\,.$\\
\item \underline{Variance Gamma case:} Let $\alpha, \beta > 0, \delta \neq 0$.
If $X$ is a Variance Gamma process with
 $X_1 \sim VG(\alpha,\beta,\delta,\mu)$ with
 $\kappa^\Lambda(z)=\mu z + \delta Log\left(\frac{\alpha}
{\alpha-\beta z - \frac{z^2}{2}}\right)\,,$
where $Log$ is again the principal value complex logarithm
defined in Section 2.
The expression of $\kappa^\Lambda(z)$ can be found in
\cite{Ka06, madan} or also
~\cite{livreTankovCont},  table IV.4.5 in the particular case  $\mu = 0$.
 In particular an easy calculation shows that we need $z \in \mathbb{C}$ 
such that $Re(z) \in ]-\beta-\sqrt{\beta^2+2\alpha},-\beta+
\sqrt{\beta^2+2\alpha}[$ so that $\kappa^\Lambda(z)$ is well-defined
so that
$$ D_\Lambda =  ]-\beta-\sqrt{\beta^2+2\alpha},-\beta+
\sqrt{\beta^2+2\alpha}[ + i \R.$$
\end{enumerate}
\end{remarque}

\begin{remarque} \label{R5.88}
We come back to the examples introduced in Remark~\ref{remark47Bis}.
In all the three cases, 
Hubalek et alia Condition  is verified if $2 \in D$.
 This happens in the following situations:
\begin{enumerate}
\item always in the Poisson case;
\item if $\Lambda=X$ is a NIG process and if $2\le\alpha-\beta\,;$
\item if $\Lambda=X$ is a VG process and if $2<-\beta+\sqrt{\beta^2+2\alpha}\,.$
\end{enumerate}
\end{remarque}
Theorem~\ref{mainthm} allows to re-obtain the
 results stated in~\cite{Ka06}.

\begin{remarque}
\label{remarquePoisson}
If X is a Poisson process with parameter $\lambda > 0$
 then the quadratic error is zero. In fact, 
\begin{eqnarray*}
\kappa^\Lambda(z) &=& \lambda (\exp(z) -1)) \ ,\quad
\rho_t(y,z) = \lambda t(\exp(y) - 1)(\exp(z) -1) \\
\gamma(z,t) &=&  \frac{\kappa^\Lambda (z+1)-
\kappa^\Lambda(z)-\kappa^\Lambda(1)}{\kappa^\Lambda(2)-
2\kappa^\Lambda(1)} t =
\frac{\exp(z) -1}{e-1}
\end{eqnarray*}
imply that $\beta(y,z,t) = 0$ for every $y,z \in \C, t \in [0,T]$.

Therefore $J_0(y,z,t)\equiv0$. In particular all the options of type~(\ref{FORM})  are perfectly hedgeable.
\end{remarque}

\subsection{About some singular non-stationary models}
Here, we consider some \textit{singular} models, in the sense that the cumulant generating function of the log-price process is not 
absolutely continuous with respect to (a.c. w.r.t.) Lebesgue measure. 
More precisely, let $(W_t)$ be a standard Brownian motion. A classical approach to model the volatility clustering effect consists in introducing the notion of \textit{trading  time} (as opposed to the real time) which accelerates or slows down the price process depending on the activity on the market. This virtual time is represented by  a change of time  $(\tau_{t})_{t\geq 0}$ and the log-price is then constructed by subordination  
i.e.  $X_t=W_{\tau(t)}$. 
Now, if the change of time $\tau$ is \textit{singular},
 then it can be proved that the log-price process $X$ is also \textit{singular}. \\ 
This typically happens when  the change of time
 $\tau$, is obtained as the cumulative distribution function of a deterministic positive multifractal measure $d\tau(t)=d\psi(t)$, singular w.r.t. Lebesgue measure. 
%
Multifractal measures
were introduced  in the physical sciences to model turbulent flows~\cite{Mandelbrot72}. 
%
More recently, in~\cite{Calvet97}, the authors used this construction precisely for modeling financial volatility. But their model, the \textit{Multifractal Model of Asset Returns} (MMAR),  relies on  a random (and not deterministic) multifractal
measure and  is hence beyond the framework of this paper. \\
Below, we consider two examples of {\it singular}
 non-stationary log-price models based on such (deterministic or random)
 {\it singular} changes of time. 

\begin{enumerate}
\item \underline{Deterministic change of time (log-Gaussian continuous process): }
Let us consider the log-price process $X$ such that $X_t=W_{\psi(t)}$, where $\psi:\mathbb{R}_+\rightarrow\mathbb{R}_+$ is a
strictly increasing function, including the pathological case where
$\psi^{'}_t=0$ a.e. 
For $z \in D = \C$, we have
$ \mathbb{E}[e^{z X_t}]=\mathbb{E}[e^{z W_{\psi(z)}}]=e^{\kappa_t(z)}=e^{\frac{z^2}{2}\psi(t)},$
so that 
$\kappa_t(z)=\frac{z^2}{2}\psi(t)\ ,\quad
 \rho_t =
\psi(t).$
Notice that $d\kappa_t(z)$ is not necessarily a. c. w.r.t. Lebesgue
measure and that this is verified as soon as $d\psi(t) \ll dt$.
Assumption~\ref{HypD} 1. is verified since $\psi$ is strictly increasing;
Assumption~\ref{HypD} 2., Assumption~\ref{Hyp1bis}
and  Assumption~\ref{Hyp1} are verified since $D = \shd = \C$ and
$\frac{d\kappa_t(z)}{d\rho_t} = \frac{z^2}{2}$ is continuous.
Consequently all the conditions to apply Theorem~\ref{mainthm} are
satisfied and
\begin{eqnarray*}\gamma(z,t)=z\ ,\quad
  \eta(z,t)=\frac{\psi(t)}{2}(z^2-z) 
\quad \textrm{and}\quad \lambda_t \equiv \frac{1}{2}.\end{eqnarray*}

%
\begin{remarque}
\label{remarqueGaussian}
Calculating $\beta(y,z,t)$ in~\eqref{eq:Beta}, we find $\beta\equiv 0$. Therefore here also the quadratic error is zero. This confirms the fact that the market is \textbf{complete}, at least for the considered class of options.
\end{remarque}
\item \underline{Random change of time: }
  Let $(\theta_t)_{t\geq 0}$ denote an increasing L\'evy process such that $\theta_1$ follows an Inverse Gaussian distribution with parameters $\delta>0$ and  $\gamma>0$.  Now, let us consider $Y$ the process 
  such that $Y_t=\mu t +\beta \theta(t)+ W_{\theta(t)}$, for all $t\in[0,T]$, with  $\beta\,,\mu \in \mathbb{R}$. Then one can prove that $Y$ is a NIG L\'evy process with $Y_1\sim NIG(\alpha =\sqrt{\gamma^2+\beta^2}, \beta, \delta, \mu)$. Finally, let us consider the log-price process $X$ such that 
$X_t=W_{\tau_{t}}$, where $\tau_t=\theta_{\psi(t)}$ and $\psi$ is the cumulative distribution of a deterministic multifractal measure on $[0,T]$. 
 Hence, the cumulant generating function of $X_t$ is singular w.r.t. Lebesgue measure and is given by  $\kappa_t(z)=[\mu z + \delta
 (\gamma_0-\gamma_z)]\psi(t)$ with 
 $\gamma_z=\sqrt{\alpha^2-(\beta+z)^2}$, for all $z\in
D:= D_{X_t} = [-\alpha-\beta,\alpha-\beta]+i\mathbb{R}\,.$
\end{enumerate}


\section{Application to Electricity}
\label{sec:elec}

\subsection{Hedging electricity derivatives with forward contracts}

Because of non-storability of electricity, no dynamic hedging strategy
can be performed on the spot market. Hedging instruments for electricity
derivatives are then futures or forward contracts.
For simplicity, we will assume that interest rates are deterministic and zero so that futures prices are equivalent to forward prices. 
 The value of a forward contract offering the fixed price $F_0^{T_d}$ at time $0$ for delivery of 1MWh at time $T_d$ is by definition of the forward price,  $S^{0,T_d}_0=0$. Indeed, there is no cost to enter at time $0$ the forward contract with the current market forward price $F_0^{T_d}$.  
Then, the value of the same forward contract $S^{0,T_d}$ at time $t\in
[0,T_d]$ is deduced by an argument of Absence of (static) Arbitrage as 
$S^{0,T_d}_t = e^{-r(T_d-t)}(F_t^{T_d} - F_0^{T_d})$. Hence,
 the dynamics of the hedging instrument $(S^{0,T_d}_t)_{0\leq t\leq T_d}$ is directly related (for deterministic interest rates) to the dynamics of forward 
prices $(F_t^{T_d})_{0\leq t\leq T_d}$.   
Consequently, in the sequel, when considering hedging on electricity markets, we will always suppose that the underlying is a forward contract   $(S^{0,T_d}_t)_{0\leq t\leq T_d}$ and we will focus on the dynamics of forward prices.

\subsection{Electricity price models for pricing and hedging application}
Observing  market data, one can notice two main stylized features of 
electricity forward prices: 
\begin{itemize}
	\item Volatility term structure of forward prices: the volatility increases when the time to maturity decreases. Indeed, when the delivery date approaches, the flow of relevant information affecting the balance between electricity supply and demand increases and causes great variations in the forward prices. This maturity effect is usually referred to
as the \textit{Samuelson hypothesis}, it was first studied in~\cite{Sam65} and can be observed on Figure~\ref{fig:VolTerm}, in the case of electricity futures prices. 
	\item Non-Gaussianity of log-returns: log-returns can be
	  considered as Gaussian for long-term contracts but begin to show heavy tails for short-term contracts.  
\end{itemize}
\begin{figure}[htbp]
\begin{center}
  \epsfig{figure=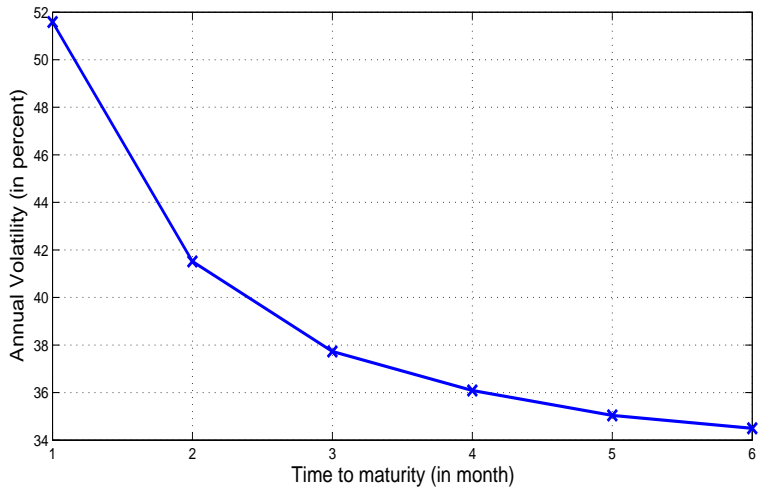,height=7cm, width=12cm}
\end{center}
\vspace{-.5cm}
\caption{{\small Volatility of electricity \textit{Month-ahead futures prices} w.r.t. to the time to maturity estimated on the French Power market in 2007. }}
\label{fig:VolTerm}
\end{figure}
Hence, a challenge is to be able to describe with a single model, both 
the non-Gaussianity on the short term and the volatility term structure of the forward curve. 
One reasonable attempt to do so is to consider the exponential L\'evy
 factor model, proposed in  
\cite{BS03} 
or \cite{oudjaneCollet}. The forward price given at time $t$ for delivery at time $T_d\geq t$, denoted $F_t^{T_d}$ is then modeled by a $p$-factors model, such that 
\begin{equation}
\label{eq:elec}
F_t^{T_d}=F_0^{T_d}\exp(m_t^{T_d}+\sum_{k=1}^pX_t^{k,{T_d}})\ ,\quad \textrm{for all}\ t\in [0,{T_d}]\ , \textrm{where}
\end{equation}
\begin{itemize}
	\item $(m_t^{T_d})_{0\leq t\leq {T_d}}$ is a real deterministic trend; 
	\item for any $k=1,\cdots p$, $(X_t^{k,{T_d}})_{0\leq t\leq {T_d}}$ is such
      that $X_t^{k,{T_d}}=\int_0^t \sigma_k e^{-\lambda_k ({T_d}-s)}
      d\Lambda^k_s$,
 where $\Lambda=(\Lambda^1,\cdots,\Lambda ^p)$ is a L\'evy process 
on $\mathbb{R}^d$, with
 $\mathbb{E}[\Lambda^k_1]=0$ and $Var[\Lambda^k_1]=1$; 
	\item $\sigma_k>0\,,\,\lambda_k \geq0\ ,$ are called respectively the \textit{volatilities} and the \textit{mean-reverting rates}. 
\end{itemize}
Hence,  forward prices are given as exponentials of 
additive processes with \textit{non-stationary increments}. 
%
In practice, we consider the case of a one or a two factors model ($p=1$
 or $2$), where the first factor $X^1$ is a non-Gaussian 
additive process and the second factor $X^2$ is a Brownian motion with $\sigma_1\gg \sigma_2$. 
Notice that this kind of model was originally developed and studied in
details for interest rates in~\cite{Raible98}, as an extension of
 the Heath-Jarrow-Morton model where the Brownian motion has
 been replaced by a general L\'evy process. 
\\
Of course, this modeling procedure~(\ref{eq:elec}), implies
incompleteness of the market. Hence, if we aim at pricing and hedging a
European call on a forward with maturity $T\leq T_d$, it won't be possible,
in general, to hedge perfectly the payoff $(F^{T_d}_T-K)_+$ with a hedging
portfolio of forward contracts.
Then, a natural approach could consist in looking for the variance
optimal initial capital and hedging portfolio. In this framework, the results of
Section~\ref{sec:expPII} generalizing the results of Hubalek \& al
in~\cite{Ka06} to the case of non stationary additive process can be
useful.

\subsection{The non Gaussian two factors model}

To simplify let us forget the superscript $T_d$ denoting the delivery period (since we will consider a fixed delivery period).
 We suppose that the forward price $F$ follows the two factors model
\begin{equation}
\label{eq:elec:2F}
F_t=F_0\exp(m_t+X^1_t+X^2_t)\ ,\quad \textrm{for all}\ t\in [0,T_d]\ ,
 \textrm{where}
\end{equation}
\begin{itemize}
\item $m$ is a real deterministic trend starting at $0$. 
It is supposed to be absolutely continuous w.r.t.  Lebesgue;
	\item  $X^1_t=\int_0^t \sigma_s e^{-\lambda (T_d-u)} d\Lambda_u$,
      where $\Lambda$ is a 
 L\'evy process on
 $\mathbb{R}$ with $\Lambda$ following a Normal Inverse Gaussian (NIG) 
distribution
or a Variance Gamma  (VG) distribution. Moreover, we will assume that
  $\mathbb{E}[\Lambda_1]=0$ and $Var[\Lambda_1]=1$; 
	\item $X^2=\sigma_l W$ where $W$ is  a standard Brownian motion 
on $\mathbb{R}$;
\item $\Lambda$ and $W$ are independent;
	\item $\sigma_s$ and $\sigma_l$ standing respectively for the short-term volatility and long-term volatility. 
\end{itemize}
%

\subsection{Verification of the assumptions}

The result below helps to extend Theorem~\ref{mainthm} to the case where
$X$ is a finite sum of independent  semimartingale additive processes,
each one verifying Assumptions~\ref{HypD},  \ref{Hyp1bis} and~\ref{Hyp1} 
for a given payoff $H = f(s_0 e^{X_T})$.
\begin{lemme}\label{LSum}
Let $X^1, X^2$ be two independent  semimartingale
additive processes with cumulant
generating functions $\kappa^i$ and related domains  
$D^i, \shd^i, i = 1,2$ characterized in Remark~\ref{CumGenPAI}
and \eqref{Eqshd}. Let $f: \C \rightarrow \C$ of the form~(\ref{FORM}). \\
For $X = X^1 + X^2$ with related domains $D, \shd$
and cumulant generating function $\kappa$, we have  the following.
\begin{enumerate}
\item $D = D^1 \cap D^2$.
\item $\shd^1 \cap \shd^2 \subset \shd$.
\item  If $X^1, X^2$ verify
Assumptions~\ref{HypD},  \ref{Hyp1bis} and~\ref{Hyp1},
then $X$ has the same property.
\end{enumerate}
\end{lemme}
\begin{proof} 
Since $X^1, X^2$ are independent and taking into account
Remark~\ref{CumGenPAI} we obtain 1. and 
$  \kappa_t(z) = \kappa_t^1(z) + \kappa^2(z), \ \forall z \in D.$
 We denote by $\rho^i, i = 1,2$, the reference variance
 measures defined
in Remark~\ref{remarque313}. Clearly 
$\rho = \rho^1 + \rho^2$ and  $d\rho^i \ll  d\rho$
with $\Vert \frac{d\rho^i}{d\rho} \Vert_{\infty} \le 1$.\\
If $z \in \shd^1 \cap \shd^2$, we can write 
\begin{eqnarray*}
\int_0^T\left|\frac{d\kappa_t(z)}{d\rho_t}\right|^2d\rho_t&\leq&2\int_0^T\left|\frac{d\kappa^1_t(z)}{d\rho^1_t}\frac{d\rho^1_t}{d\rho_t}\right|^2d\rho_t 
+2\int_0^T\left|\frac{d\kappa^2_t(z)}{d\rho^2_t}\frac{d\rho^2_t}{d\rho_t}\right|^2d\rho_t \\
&=&2\int_0^T\left|\frac{d\kappa^1_t(z)}{d\rho^1_t}\right|^2\frac{d\rho^1_t}{d\rho_t}d\rho^1_t 
+2\int_0^T\left|\frac{d\kappa^2_t(z)}{d\rho^2_t}\right|^2\frac{d\rho^2_t}{d\rho_t}d\rho^2_t \\
&\leq& 2\left(\int_0^T\left|\frac{d\kappa^1_t(z)}{d\rho^1_t}\right|^2d\rho^1_t+\int_0^T\left|\frac{d\kappa^2_t(z)}{d\rho^2_t}\right|^2d\rho^2_t\right)\ .
\end{eqnarray*}
This concludes the proof of $\shd^1 \cap \shd^2 \subset \shd$ and
therefore of the of
Point 2.  \\
Finally Point 3. follows then by inspection.
\end{proof}

With the two factors model, the forward price $F$  is then given
 as the exponential of an additive process,
 $X$, such that for all $t\in [0,{T_d}]$, 
\begin{equation}
\label{eq:X}
X_t=m_t+X^1_t+X^2_t=m_t+\sigma_s 
\int_{0}^t e^{-\lambda ({T_d}-u)}d\Lambda_u+\sigma_{l}W_t\ .
\end{equation}
For this model, we formulate the following assumption.
\begin{Hyp}  \label{AssuElec}
\begin{enumerate}
\item $ 2 \sigma_s \in D_\Lambda$.
\item If $\sigma_l = 0$, we require $\Lambda$ not to have deterministic
increments.
\item $f: \C \rightarrow \C$ is of the type~(\ref{FORM})
 fulfilling \eqref{E4000}.
\end{enumerate}
\end{Hyp}
\begin{propo} \label{PElec61}
\begin{enumerate} 
\item The cumulant generating function of $X$ defined by~(\ref{eq:X}), 
 $\kappa:[0,{T_d}]\times D\rightarrow \mathbb{C}$ is such that 
for all $z \in    D_\Lambda(\sigma_s)$ and
  for all $t\in [0,T_d]$, 
\begin{equation} \label{99}
\kappa_t(z)=z m_t+\frac{z^2\sigma^2_{l}t}{2}+\int_0^t \kappa^\Lambda(z\sigma_s e^{-\lambda ({T_d}-u)})du \ .
\end{equation}
In particular for fixed $z \in D_\Lambda(\sigma_s)$,
 $t \mapsto \kappa_t(z)$ is absolutely continuous w.r.t. 
 Lebesgue  measure.
\item  Under Assumption \ref{AssuElec}, Assumptions~\ref{HypD},
~\ref{Hyp1bis} and  
\ref{Hyp1} are fulfilled.
\end{enumerate}
\end{propo}
\begin{proof}
We set $\tilde X^2 = m + X^2$. We observe that 
$ D^2 = \shd^2 = \C, \quad \kappa^2_t(z) = \exp(z m_t + z^2 \sigma_l^2 
\frac{t}{2}). $
We recall that $\Lambda$ and $W$ are independent so that
$\tilde X^2$ and $X^1$ are independent.
For clarity, we only write the proof under the hypothesis that
$\Lambda$ has no deterministic increments,
the general case could be easily adapted.
$X^1$ is a process of the type studied at Section~\ref{ExpWiener};
it verifies Assumption~\ref{HypWiener} and  $D_\Lambda(l)$
contains $D_\Lambda(\sigma_s)$.

According to Proposition~\ref{propo411},
Remark~\ref{remarque412} and  \eqref{E4000}
it follows that Assumptions~\ref{HypD},~\ref{Hyp1bis} and  
\ref{Hyp1}
are verified for $X^1$. 
Both statements 1. and 2. are now a consequence of 
Lemma~\ref{LSum}. 
\end{proof}

The solution to the mean-variance problem is provided by Theorem~\ref{mainthm}.
\begin{thm} We suppose Assumption~\ref{AssuElec}.
The variance-optimal capital $V_0$ and the variance-optimal
 hedging strategy $\varphi$,
 solution of the minimization problem $(\ref{problem2})$, are given by
 Theorem~\ref{mainthm} and Theorem~\ref{propo310}, 
Proposition~\ref{lemme38}
together with the expressions given below:
%
\begin{eqnarray*}
\widetilde{l}_t :&=&\sigma_s e^{-\lambda ({T_d}-t)},\\
\gamma(z,t) :&=&\frac{z\sigma^2_l+\kappa^\Lambda
((z+1)\widetilde{l}_t)-\kappa^\Lambda(z\widetilde{l}_t)-
\kappa^\Lambda(\widetilde{l})_t}
{\sigma^2_l+\kappa^\Lambda(2\widetilde{l}_t)-
2\kappa^\Lambda(\widetilde{l}_t)}, \\
\eta(z,t):&=&\left [ zm_t+\frac{z^2\sigma^2_l}{2}+
\kappa^\Lambda(z\widetilde{l}_t)
-\gamma(z,t) \big
(m_t+\frac{\sigma^2_l}{2}+\kappa^\Lambda(\widetilde{l}_t)\big )
\right]\,dt\ ,
\\
\lambda_t &=&\frac{m_t+\frac{\sigma_l^2}{2}+
\kappa^\Lambda(\widetilde{l}_t)}{\sigma^2_l+
\kappa^\Lambda(2\widetilde{l}_t)-2\kappa^\Lambda(\widetilde{l}_t)}. 
\end{eqnarray*}
\end{thm}
\begin{remarque} 
\label{R65Levy}
Previous formulae are practically exploitable numerically.
 The last condition to be checked is 
\begin{eqnarray} \label{E731}
2 \sigma_s \in D_\Lambda.
\end{eqnarray}
%
\begin{enumerate}
\item $\Lambda_1$ is a Normal Inverse Gaussian random variable; if $ \sigma_s\leq \frac{\alpha-\beta}{2}$ then $(\ref{E731})$ is verified.
\item $\Lambda_1$ is a Variance Gamma random variable then $(\ref{E731})$ is verified; if for instance $\sigma_s<\frac{-\beta+\sqrt{\beta^2+2\alpha}}{2}\ .$
\end{enumerate}
\end{remarque}

\section{Simulations}
We are interested in comparing, in simulations, the Variance Optimal (VO) strategy to
 the Black-Scholes (BS) strategy when hedging a European call, with payoff $(S_T-K)_+$, 
on an underlying stock with log-prices $X_t=\log(S_t)$ that have independent but non Gaussian increments. More precisely, we assume that the underlying is an electricity forward contract 
$S_t=S^{0,T_d}_t=e^{-r(T_d-t)}(F_t^{T_d}-F_0^{T_d})$ 
with  delivery date $T_d$ equal to the maturity of the call
$T_d=T$. \\
First, we consider the case where the log-price
process $X$ is an exponential of a L\'evy process, 
continuing the analysis of~\cite{Ka06}, then we consider the non 
stationary case. 
We make use of different simulated data according to the underlying model, 
stationary in one case, non stationary in the second one.

Our simulations investigate two features which were not considered in
\cite{Ka06} (even in the stationary case): first the robustness of the
BS hedging strategy w.r.t.
 the underlying price model, second the sensitivity of the continuous
 VO strategy w.r.t.  to the discreteness of the trading dates. \\
 The VO strategy knows the real incomplete price model (with the real values of parameters) whereas the BS strategy assumes (wrongly) a log-normal price model (with the real values of mean and variance).  Of course, the VO strategy is by definition optimal, w.r.t. the quadratic norm. However, both strategies (VO and BS) are implemented in discrete time, hence our goal is precisely to analyze the hedging error outside of the theoretical framework of a continuously rebalanced portfolio. Moreover, we are interested in interpreting quantitatively the differences between both strategies w.r.t. to some characteristics such as the underlying log-returns distribution or the number of trading dates. 
%
\\
The time unit is the year and the interest rate is zero in all our simulations. The initial value of the underlying is $s_0=100$ Euros.  The maturity of the option is $T=0.25$ i.e. three months from now. 

\subsection{Exponential L\'evy}

In this subsection, we simulate 
 the log-price process $X$ as a NIG L\'evy process with 
$X_1\,\sim\, NIG(\alpha,\beta,\delta,\mu)$. 
Five different sets of parameters for the NIG distribution have been considered, going from the case of \textit{almost Gaussian} returns corresponding to standard equities, 
 to the case of \textit {highly non Gaussian} returns. 
The standard set of parameters 
is  estimated on the \textit{Month-ahead} \textit{base} forward prices of the French Power market in 2007:  
\begin{equation}
\label{eq:para;levy}
\alpha=38.46 \,,\  \beta= -3.85\,,\ \delta = 6.40\,,\  \mu= 0.64 \ .
\end{equation}
Those parameters imply a zero mean, a standard deviation of $41\%$, a
skewness (measuring the asymmetry) of $-0.02$ and an excess kurtosis
(measuring the \textit{fatness} of the tails) of $0.01$. The other sets
of parameters are obtained by multiplying the parameter  $\alpha$ by a
coefficient $C$, ($\beta,\delta,\mu$) being such that the first three
moments are unchanged. Note that when $C$ grows to infinity the tails of
the NIG distribution get closer to the tails of the Gaussian
distribution. For instance, Table~\ref{tab:kurtosis} shows how the
excess kurtosis (which is zero for a Gaussian distribution) is modified
with the five values of $C$ chosen in our simulations. 
\begin{table}[htbp]
\begin{center}
\begin{tabular}{|c||c|c|c|c|c|c|}
\hline
 \textrm{Coefficient} & $C=0.08$ & $C=0.14$ & $C=0.2$ & $C=1$ &$C=2$ \\
\hline
\hline
 $\alpha$&  $3.08$&  $5.38$&  $7.69$&  $38.46$&  $76.92$ \\
\hline
 Excess kurtosis&  $1.87$&  $0.61$&  $0.30$&  $0.01$&  $4. \,10^{-3}$ \\
\hline
\end{tabular}
\end{center}
\vspace{-.1cm}
\caption{{\small Excess kurtosis of $X_1$ for different values of $\alpha$,  $(\beta, \delta, \mu)$ insuring the same three first moments.}}
\label{tab:kurtosis}
\end{table}
%

\subsubsection{Strike impact on the initial capital and the hedging ratio}
Figure~\ref{fig:strike} shows the initial capital (on the left graph) and the initial hedge ratio (on the right graph) produced by the VO 
and the BS strategies as functions of the strike, for three different sets of parameters $C=0.08\,,\ C=1\,,\ C=2$. We consider $N=12$  trading dates, which corresponds to operational practices on electricity markets, for an option expiring in three months. One can observe that BS results are very similar to VO results for $C\geq 1$ i.e. for \textit{almost Gaussian returns}.  
However, for small values of $C$, for $C=0.08$,  corresponding to highly non Gaussian returns,  BS approach under-estimates \textit{out-of-the-money} options and over-estimates \textit{at-the-money} options (for $K=99$ Euros the BS  initial capital is equal to $8.65$ Euros i.e. $122\%$ of the VO initial capital, while for $K=150$, it vanishes to $23$ Cents i.e.  only  $57\%$ of the  VO initial capital). 
%
\begin{figure}[htbp]
\begin{center}
   \epsfig{figure=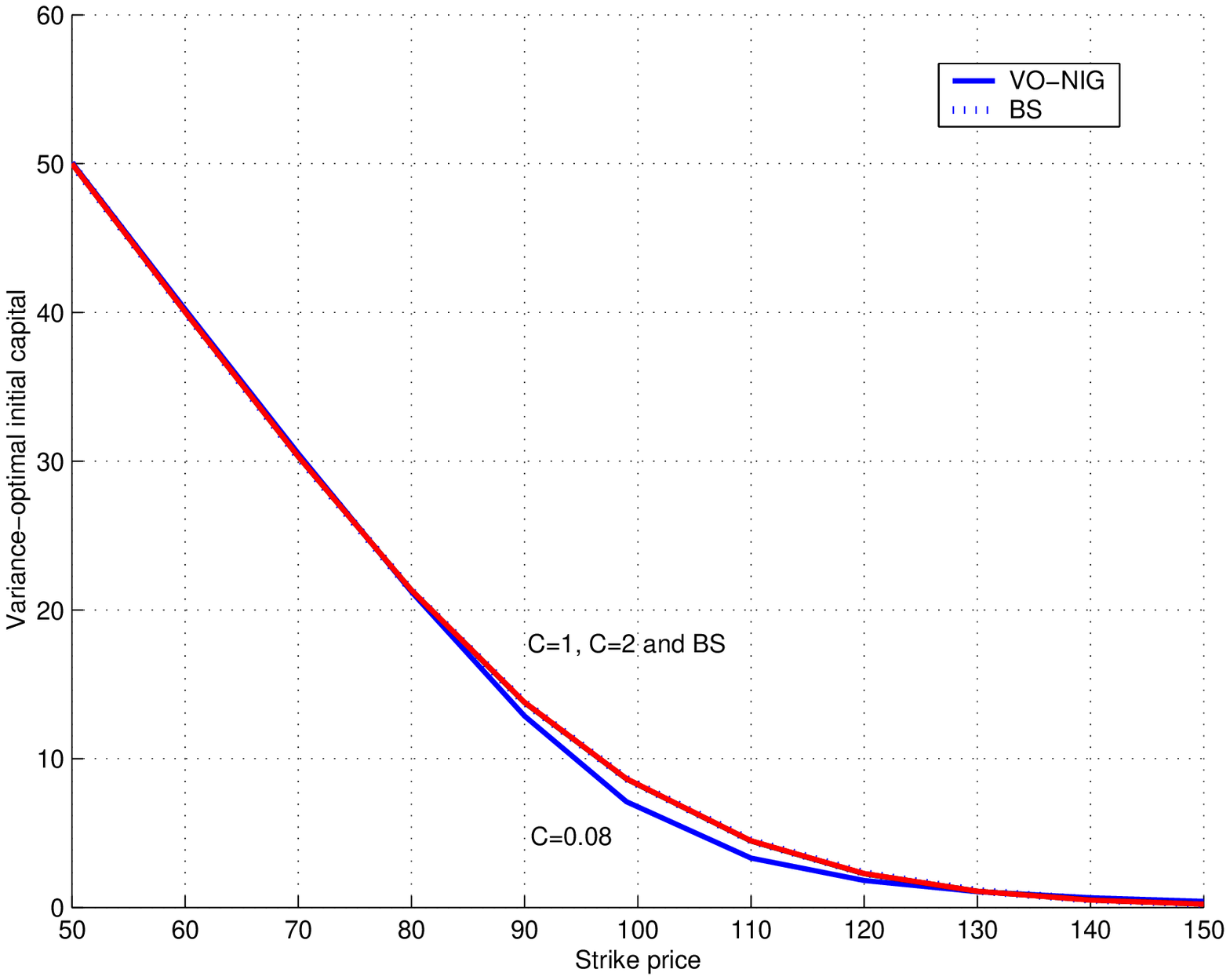,height=5cm, width=7.5cm}
   $\quad$
   \epsfig{figure=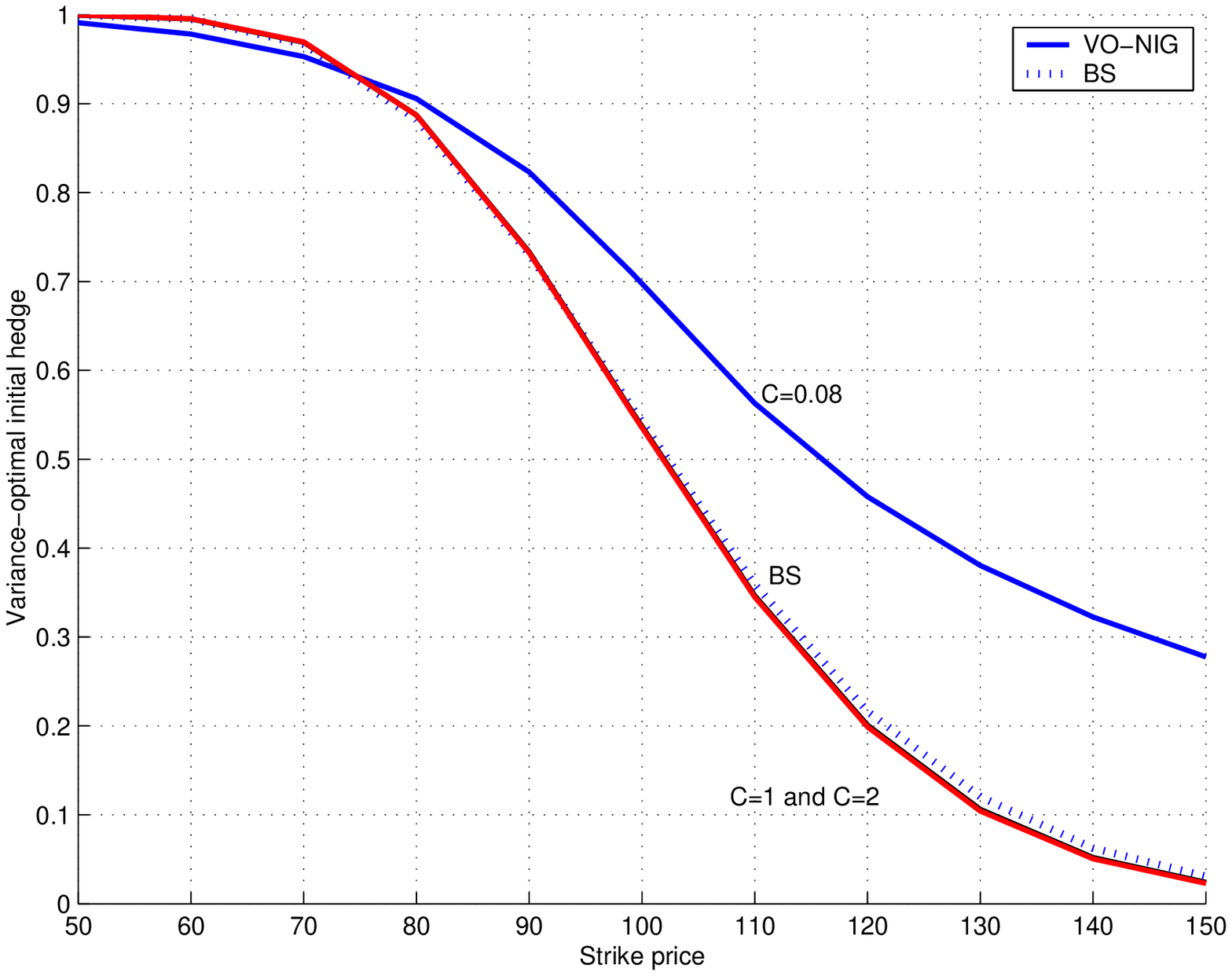,height=5cm, width=7.5cm}
\end{center}
\vspace{-.5cm}
\caption{{\small Initial capital (on the left) and hedge ratio (on the right) w.r.t. the strike, for $C=0.08\,,\ C=1\,,\ C=2$.  }}
\label{fig:strike}
\end{figure}
%

\subsubsection{Hedging error and number of trading dates}
Figure~\ref{fig:trading:dates} considers the hedging error (the
difference between the terminal value of the hedging portfolio and the
payoff) w.r.t. the number of trading dates, for a strike
$K=99$ Euros (at the money) and for five different sets of parameters $C$
given on Table~\ref{tab:kurtosis}. 
The bias (on the left graph) and standard deviation (on the right graph) of the hedging error have been estimated by Monte Carlo method on $5000$ runs. Note that we could have used the formula stated in Theorem~\ref{thm34} to compute the variance of the error, but this would have given us the limiting error which does not take into account the additional error due to the finite number of trading dates. 

In terms of standard deviation, the VO strategy seems to outperform noticeably the BS strategy, for small values of $C$ (
for $C=0.08$ the VO strategy allows to reduce $10\%$ of the standard deviation of the error).  
As expected, one can observe that the VO error converges to the BS error when $C$ increases. This is due to the convergence of  NIG log-returns to Gaussian log-returns when $C$ increases (recall that the simulated log-returns are almost symmetric). 
On Figure~\ref{fig:trading:dates}, the hedging error (both for BS and VO) decreases with the number of trading dates and seems to converge to a limiting error. Here, it is interesting to distinguish two sources of incompleteness, the \textit{rebalancing error} due to the finite number of trading dates and the \textit{intrinsic error} due to the price model incompleteness. 
For instance, one can observe that for small values of $C\leq 0.2$, even for small numbers of trading dates,  the \textit{intrinsic error} seems to be predominant so that it seems useless to increase the number of trading dates over $N\geq 12$ trading dates. 
Moreover, surprisingly one can observe that for a small number of trading dates $N\leq 12$ and for large values of $C\geq 1$,  BS seems to outperform the VO strategy, in terms of standard deviation.
 This can be interpreted as a consequence of the central limit theorem. Indeed, when the time between two trading dates increases the corresponding increments of the L\'evy process converge to a Gaussian variable. 
Similarly to the observation of \cite{DGKMP}, section 5., in term of hedging
errors, BS strategy seems to be quite close to VO strategy.
The same kind of conclusions were obtained in the discrete time setting
by \cite{AH09}.

In term of bias, the over-estimation of at-the-money options (observed
for $C=0.08$, on Figures~\ref{fig:strike}) seems to induce
a positive bias for the BS error (see Figure~\ref{fig:trading:dates}),
whereas the bias of the VO error is negligible (as expected from the
theory). 

\begin{figure}[htbp]
\begin{center}
   \epsfig{figure=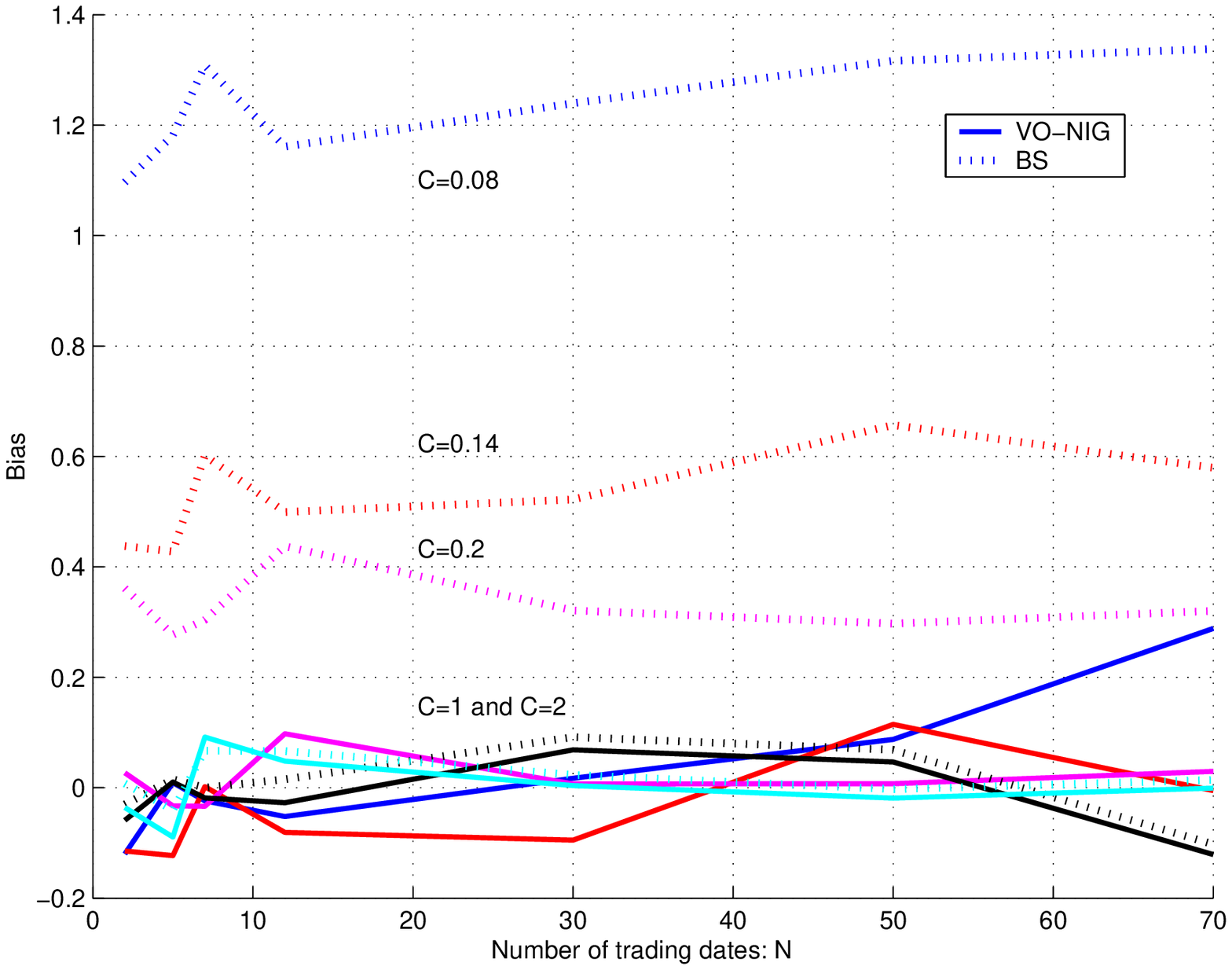,height=5cm, width=7.5cm}
   $\quad$
   \epsfig{figure=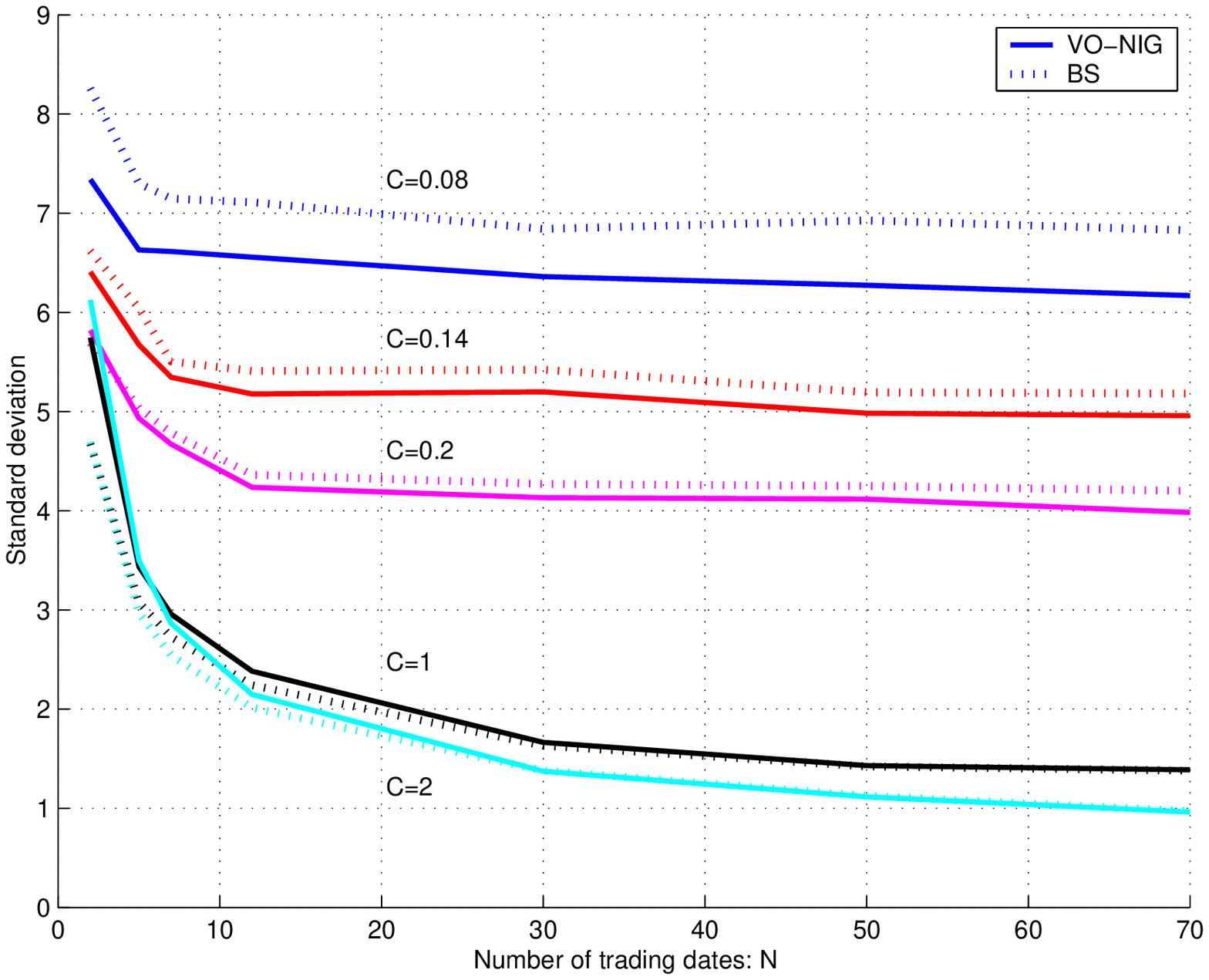,height=5cm, width=7.5cm}
\end{center}
\vspace{-.5cm}
\caption{{\small Hedging error w.r.t. the number of trading dates for different values of $C$ and for $K=99$ Euros (bias, on the left 	and standard deviation, on the right).}}
\label{fig:trading:dates}
\end{figure}
\subsection{Exponential of additive processes}
%
In this subsection, we simulate the log-price process $X$ as
 an additive process such that  
$$
X_t=\int_0^t \sigma_s e^{-\lambda(T-u)}d\Lambda_u \quad  \textrm{where $\Lambda$ is a L\'evy process with  }\quad \Lambda_1\,\sim\, NIG(\alpha,\beta,\delta,\mu)\ .
$$
The standard set of parameters $(C=1)$ for the distribution of $\Lambda_1$ is  estimated on the same data as in the previous section (\textit{Month-ahead} \textit{base} forward prices of the French Power market in 2007):  
$$\alpha=15.81 \,,\  \beta= -1.581\,,\ \delta = 15.57\,,\  \mu= 1.56 \ .$$
Those parameters correspond to  a standard and centered NIG distribution with a skewness of $-0.019$.
The estimated annual  short-term volatility and mean-reverting rate are $\sigma_s=57.47\%$ and $\lambda=3$. 
The other sets of parameters considered in simulations are obtained by multiplying parameter  $\alpha$ by a coefficient $C$, ($\beta,\delta,\mu$ being such that the first three moments are unchanged). 

The results are comparable to those obtained in the case of the L\'evy
process, on Figure~\ref{fig:trading:datesPAI}. 
%
\begin{figure}[htbp]
\begin{center}
   \epsfig{figure=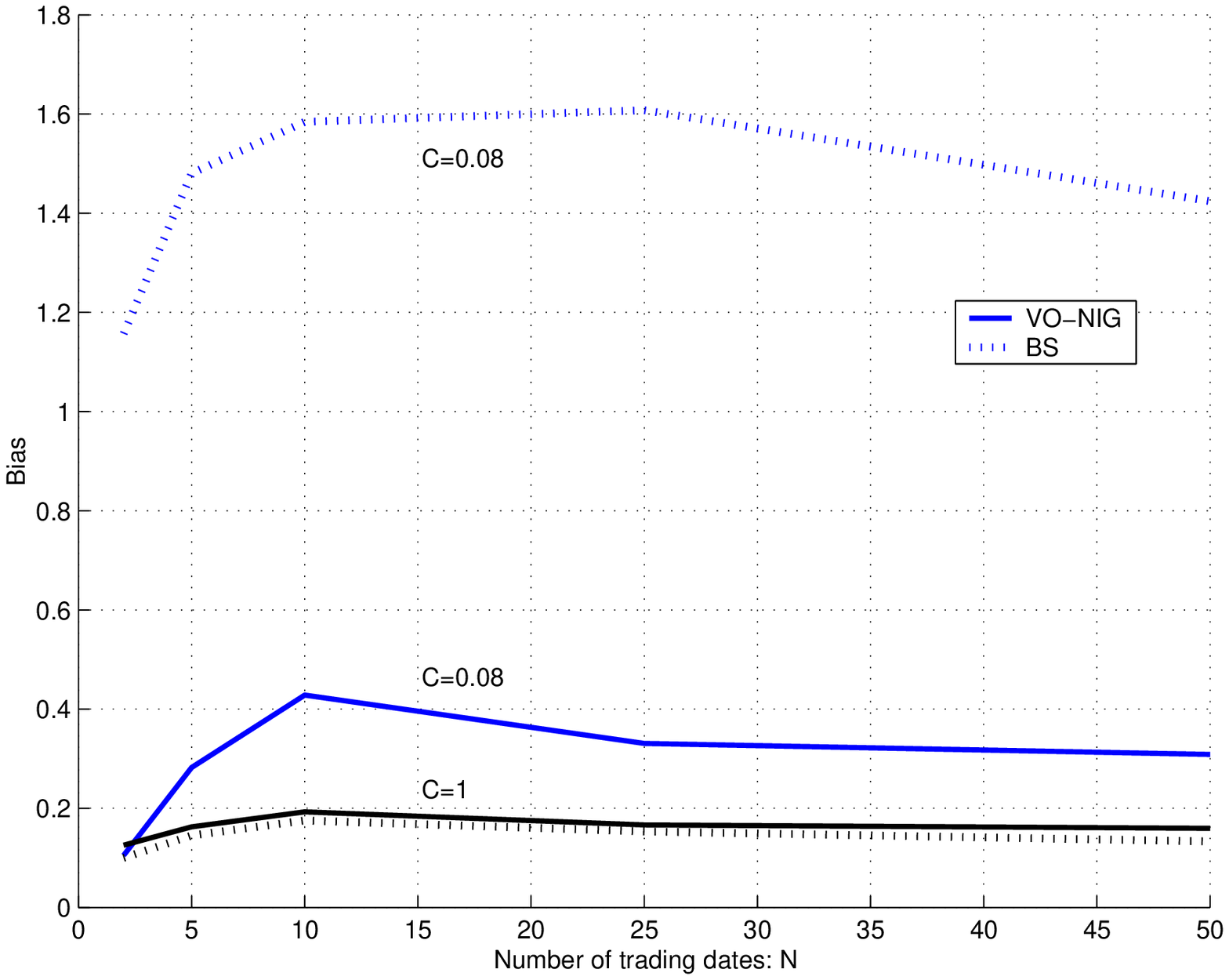,height=5cm, width=7.5cm}
   $\quad$
   \epsfig{figure=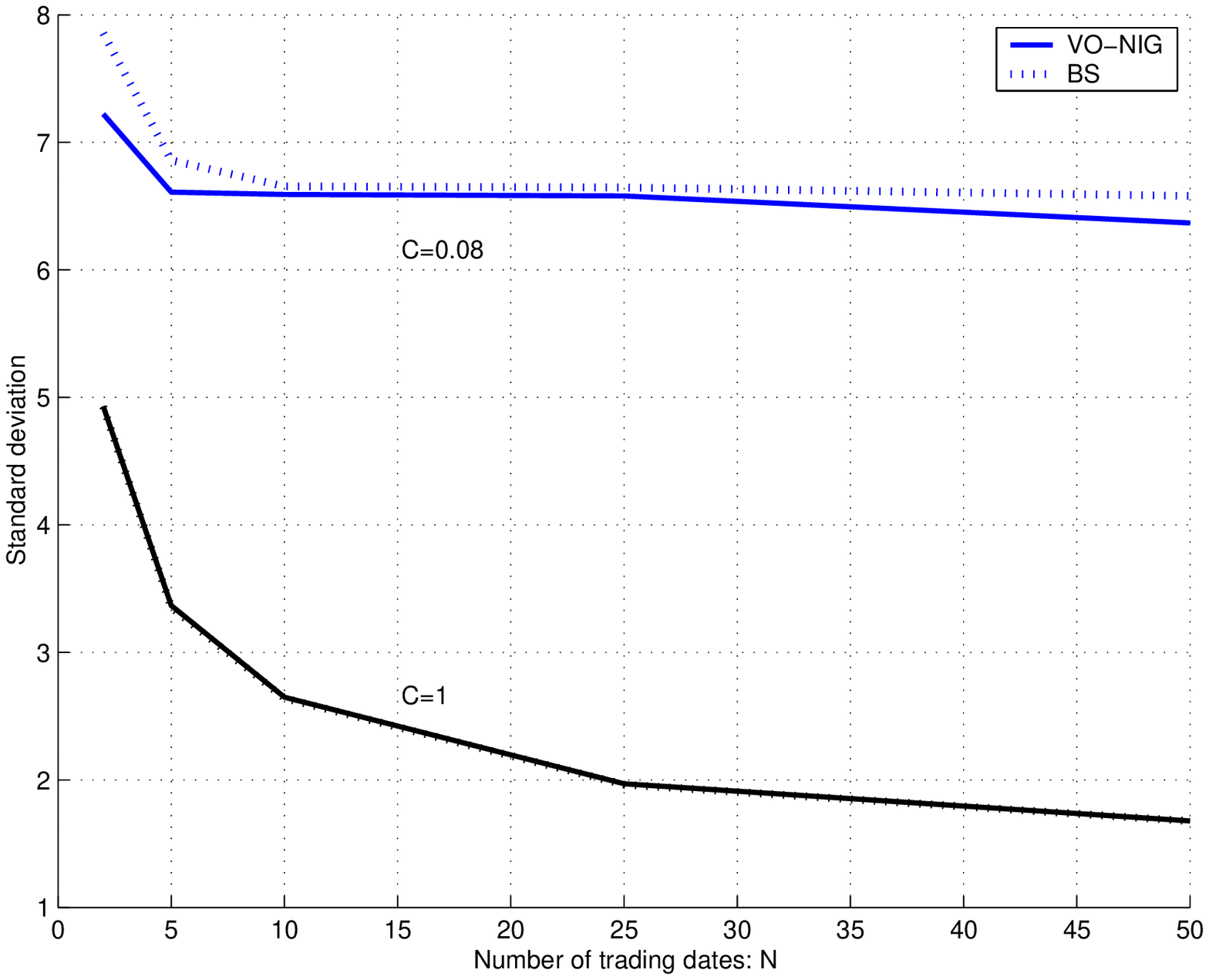,height=5cm, width=7.5cm}
\end{center}
\vspace{-.5cm}
\caption{{\small Hedging error w.r.t. the number of trading dates for $C=0.08$ and $C=1$, for $K=99$ Euros (bias, on the left 	and standard deviation, on the right).}}
\label{fig:trading:datesPAI}
\end{figure}
%
%


\bigskip
{\bf ACKNOWLEDGEMENTS:} The first named author was partially founded
by Banca Intesa San Paolo. 
The research of the third named author was partially 
supported by the ANR Project MASTERIE 2010 BLANC-0121-01. 
All the authors are grateful to F. Hubalek for useful advices to
improve the numerical computations of Laplace transforms performed in
simulations. They also thank the Referee  for
reading carefully the manuscript and making useful comments.

\bigskip


\end{document}